\documentclass[journal]{IEEEtran}

\usepackage[english]{babel} 

\makeatletter
\adddialect\l@ENGLISH\l@english
\makeatother
                        
\usepackage{cite}
\usepackage{fixltx2e}
\usepackage[cmex10]{amsmath}
\usepackage{array}
\usepackage{dsfont}
\usepackage[hyperindex]{hyperref}
\usepackage[latin1]{inputenc}                       
\usepackage[T1]{fontenc}
\usepackage{mathtools}
\usepackage{amssymb} 
\usepackage{enumerate}
\usepackage{bbm}
\usepackage[pdftex]{graphicx}
\usepackage{epsfig,syntonly}
\usepackage{verbatim,times}
\usepackage{epstopdf}
\usepackage{enumerate}
\usepackage{accents}
\usepackage{latexsym,fancyhdr,bm}
\usepackage[tight,footnotesize]{subfigure}
\usepackage{url}

\DeclarePairedDelimiter{\ceil}{\lceil}{\rceil}
\DeclarePairedDelimiter{\floor}{\lfloor}{\rfloor}

\newtheorem{defi}{Definition}

\newtheorem{teo}{Theorem}
\newtheorem{cor}{Corollary}

\newtheorem{lemma}{Lemma}

\title{Scaling Exponent of List Decoders with Applications to Polar Codes}

\author{Marco Mondelli, S. Hamed Hassani, and R\"{u}diger Urbanke
\thanks{M. Mondelli and R. Urbanke are with the School of Computer and Communication Sciences, EPFL, CH-1015 Lausanne, Switzerland
(e-mail: \{marco.mondelli, ruediger.urbanke\}@epfl.ch).

S. H. Hassani is with the Computer Science Department, ETH Z\"{u}rich, Switzerland
(e-mail: hamed@inf.ethz.ch).

}}

\begin{document}
\maketitle
\begin{abstract}
\noindent Motivated by the significant performance gains which polar codes experience under successive cancellation list decoding, their scaling exponent is studied as a function of the list size. In particular, the error probability is fixed and the trade-off between block length and back-off from capacity is analyzed. A lower bound is provided on the error probability under $\rm MAP$ decoding with list size $L$ for any binary-input memoryless output-symmetric channel and for any class of linear codes such that their minimum distance is unbounded as the block length grows large. Then, it is shown that under $\rm MAP$ decoding, although the introduction of a list can significantly improve the involved constants, the scaling exponent itself, i.e., the speed at which capacity is approached, stays unaffected for any finite list size. In particular, this result applies to polar codes, since their minimum distance tends to infinity as the block length increases. A similar result is proved for genie-aided successive cancellation decoding when transmission takes place over the binary erasure channel, namely, the scaling exponent remains constant for any fixed number of helps from the genie. Note that since genie-aided successive cancellation decoding might be strictly worse than successive cancellation list decoding, the problem of establishing the scaling exponent of the latter remains open.
\end{abstract}


\section{Introduction}

{\bf \em Error Exponent and Scaling Exponent.} While studying the error performance of a code family when transmission takes place over a binary-input memoryless output-symmetric channel ($\rm BMSC$) $W$ with Shannon capacity $C$, the parameters of interest are, in general, the rate $R$, the block length $N$, and the block error probability $P_e$. Ideally, we would like to characterize $P_e(N, R, W)$ exactly as a function of its parameters, in particular $N$ and $R$, but this is hard to achieve. A slightly easier task is to fix one
of the quantities $(P_e, N, R)$ and then to explore the trade-off
between the remaining two.

The oldest such approach is to compute the \emph{error exponent}:
we fix the rate $R$ and we are interested in the trade-off
between $P_e$ and $N$. In particular, we compute how $P_e$ behaves
when $N$ tends to infinity. For various standard classical
random ensembles (e.g., the Shannon ensemble or the Fano ensemble,
see \cite{urbanke:coding}), it is well known that $P_e$ tends to $0$ exponentially fast in the block length, i.e.,  $P_e = \Theta(e^{-\alpha N})$, for
any $0 < R < C$. For a fairly recent survey on how to determine
$\alpha$ for various such ensembles, we refer to \cite{barg:err}.

However, the error exponent gives only limited guidance for the
design of practical coding systems, since it concerns the behavior
of the error probability once it has already reached very low values.
From an engineering point of view, the following alternate analysis
proves more fruitful: fix $P_e$ and study how the block length $N$
scales with the gap from capacity $C-R$. This scaling is relevant from a practical perspective since we typically have a certain requirement on the error probability and then we are interested in
using the shortest code possible to transmit at a certain rate. As a benchmark, let us mention what is the shortest block length
that we can hope for. A sequence of works starting from \cite{dobr:math}, then \cite{strassen:asym}, and finally \cite{poly:finite} showed that the minimum possible block length $N$ required to achieve a rate $R$ with a fixed error probability $P_e$ is roughly
equal to
\begin{equation}
N \approx \frac{V(Q^{-1}(P_e))^2}{(C-R)^2},
\end{equation}
where $V$ is referred to as
channel dispersion and measures the stochastic variability of the channel relative to a deterministic channel with the same capacity. A similar asymptotic expansion is put forward in \cite{hayashi:info} by using the information spectrum method.

This type of analysis has been successfully applied to iteratively decoded LDPC ensembles
\cite{urbanke:scalingLDPC}, where it was dubbed the \emph{scaling
law paradigm}, following a terminology coming from statistical
physics: if a system goes through a phase transition as a control
parameter $R$ crosses a critical value $C$, then generically around
this point there exists a very specific scaling law. In formulae,
we say that a scaling law holds for the block error probability
$P_e(N, R, W)$ of a capacity-achieving code if there exists a
function $f$, called the \emph{mother curve}, and a constant $\mu
> 0$, called the \emph{scaling exponent}, such that
\begin{equation} \label{eq:scal}
\lim_{N\rightarrow\infty :\mbox{ }  N^{1/\mu}(C-R) = z} P_e(N, R, W) =f(z).
\end{equation}
As the block length increases, if a rate $R< C$ is fixed, then
$P_e(N, R, W) \rightarrow 0$, since the code is supposed to achieve
capacity. On the other hand, $P_e(N, R, W) \rightarrow 1$ for any
$R> C$. Equation \eqref{eq:scal} refines this basic observation,
specifying the speed at which the rate converges to capacity, if a
certain error probability is to be met: roughly speaking, the back-off
from capacity $C-R$ tends to $0$ at a speed of $N^{-1/\mu}$. According to the previous discussion, for
random ensembles the scaling exponent is
$\mu =2$ \cite{strassen:asym, poly:finite}.

{\bf \em List Decoding}.
List decoding, which was introduced independently by Elias and
Wozencraft \cite{elias:list, woz:list}, allows the receiver to
collect $L$ possible transmitted messages. An error is declared
only if the correct message does not appear in the list. 

The \emph{error exponent} of list decoding schemes has been widely studied in the literature \cite{shannon:lb, forney:expbound}, and for random coding it has been proven that the introduction of a list with finite size $L$ does not yield any change in this asymptotic regime, provided that the rate is close enough to capacity \cite{gallager:info}. Improved bounds suitable for both random and structured linear block codes have been recently investigated \cite{hof:list}. 

As concerns the \emph{scaling exponent}, for a random ensemble transmitted over a Binary Erasure Channel with erasure probability $\varepsilon$, namely a ${\rm BEC}(\varepsilon)$, it can be shown that the error probability $P_e(N, R, \varepsilon, L)$ scales as\footnote{Consider a random matrix with $NR$ rows and $N-E$ columns whose elements are independent random variables taking the values 0 and 1 with equal probability and where $E$ is a Bernoulli random variable with mean $N\varepsilon$ and variance $ N\varepsilon (1-\varepsilon)$. Then, $P_e(N, R, \varepsilon, L)$ is the probability that this matrix has rank $< NR-\log_2 L$. After some calculations and the application of Theorem 3.2.1 of \cite{Ko99}, one obtains that the dominant term in $P_e(N, R, \varepsilon, L)$ is given by ${\mathbb P}(E > N(1-R)+\log_2 L)$, which is expanded in \eqref{eq:randbnd}.} 
\begin{equation}\label{eq:randbnd}
\begin{split}
P_e(N, R, \varepsilon, L) &\approx Q\biggl(\frac{\log_2 L}{\sqrt{N\varepsilon(1-\varepsilon)}}+ \frac{\sqrt{N}(1-\varepsilon-R)}{\sqrt{\varepsilon(1-\varepsilon)}}\biggr),
\end{split}
\end{equation}
where $Q(x)= 1/\sqrt{2\pi} \int_{x}^{+\infty} \exp{(-u^2/2)}du$. Consequently, the scaling exponent remains equal to
$2$ and also the mother curve stays unchanged, namely $f(z) =
Q(z/\sqrt{\varepsilon(1-\varepsilon)})$, for any $L \in {\mathbb
N}$.  

{\bf \em Polar Codes: a Motivating Case Study}.
The present research was motivated by polar codes, which were
recently introduced by Ar{\i}kan in \cite{arikan:polar}, and that provably achieve the capacity of a large class of channels including the class of $\rm BMSC$s. The encoding complexity and the decoding complexity with the proposed successive cancellation ($\rm SC$) decoder is $\Theta(N \log
N)$ (see Sections VII and VIII of \cite{arikan:polar}).

In particular, for any $\rm BMSC$ $W$ and for any rate $R < C$, it has been proven that the block error probability
under $\rm SC$ decoding,
namely $P_e^{\rm SC}(N, R, W)$, behaves roughly as $2^{-\sqrt{N}}$
as $N$ grows large \cite{arikan:rate}. With an abuse of notation,
it is said that polar codes achieve an \emph{error exponent} of
$\alpha = 1/2$. This result has been further refined and extended
to the $\rm MAP$ decoder, showing that both $\log_2(-\log_2 P_e^{\rm
SC})$ and $\log_2(-\log_2 P_e^{\rm MAP})$ scale as $\log_2 (N)/2+\sqrt{\log_2 (N)}/2\cdot
Q^{-1}(R/C) + o(\sqrt{\log_2 (N)})$ for any fixed rate strictly less than
capacity \cite{hamed:rated}.

However, when we consider rates close to capacity, simulation
results show that large block lengths are required in order to achieve
a desired error probability. Therefore, it is interesting to
explore the trade-off between rate and block length when the error
probability is fixed, i.e., to consider the \emph{scaling approach}.

In \cite{korada:scaling}, the authors provide strong evidence that
a lower bound to the block error probability under $\rm SC$ decoding
satisfies a scaling law \eqref{eq:scal}. In particular, a proper
scaling assumption yields the inequalities
\begin{equation}
\label{eq:boundscal}
f(N^{\frac{1}{\mu}}(C-R)) \le P_e^{\rm SC}(N, R, W) \le N^{1-\frac{1}{\mu}}F(N^{\frac{1}{\mu}}(C-R)).
\end{equation}
For transmission over the $\rm BEC$, the asymptotic behavior of $f$ and $F$ for small values of their argument is provided, as well as an estimation for the \emph{scaling exponent} is given, namely $\mu \approx 3.627$. Therefore, compared to random and LDPC codes, which have a scaling exponent of $2$ for a large class of parameters and channel models \cite{urbanke:scalingLDPC}, polar codes require larger block lengths to achieve the same rate and error probability. In addition, numerical results show that $P_e^{\rm SC}(N, R, W)$ is extremely close to the upper bound \eqref{eq:boundscal}, provided that such a value is not too big ($< 10^{-1}$ suffices).

For a generic $\rm BMSC$, taking as a proxy of the error probability
the sum of the Bhattacharyya parameters, the scaling exponent is lower bounded by $3.553$ \cite{HAU14} and upper
bounded by 5.77 \cite{GB13}. Furthermore, it is conjectured that the lower bound on $\mu$ can be increased up to 3.627, which is the
value for the $\rm BEC$.

In order to improve the finite-length performance of polar codes,
a successive cancellation list ($\rm SCL$) decoder was proposed in
\cite{vardy:listpolar}. Empirically, the usage of $L$ concurrent
decoding paths yields a significant improvement in the achievable
error probability. Hence, it is interesting to study the behavior
of the scaling exponent in the context of list decoding.

{\bf \em Contribution of the Present Work.} This paper studies whether the scaling exponent is improved by adding a finite list to a decoding algorithm.

The main result concerns the behavior of the $\rm MAP$ decoder: we show that the scaling exponent does not improve for any finite list size, for any $\rm BMSC$ $W$, and for any family of linear codes whose minimum distance grows arbitrarily large when the block length tends to infinity. Proving that the minimum distance of polar codes is unbounded in the limit $N \rightarrow +\infty$, we deduce that these conclusions also hold for polar codes. In particular, by means of a \emph{Divide and Intersect} (DI) procedure, we show that the error probability of the $\rm MAP$ decoder with list size $L$, namely $P_e^{\rm MAP}(N, R, W, L)$, is lower bounded by $P_e^{\rm MAP}(N, R, W, L=1)$ raised to an appropriate power times a suitable constant, both of which depend only on $L$. As a result, we see that list decoding has the potential of significantly improving the involved constants, but it does not change the scaling exponent.

Furthermore, we consider genie-aided $\rm SC$ decoding of polar codes for transmission over the $\rm BEC$. This decoder runs the $\rm SC$ algorithm and it is allowed to ask the value of a certain bit to the genie for a maximum of $k$ times. The $k$-genie-aided $\rm SC$ decoder performs slightly worse than the $\rm SCL$ decoder with list size $2^k$, but it is easier to analyze. In this paper, we show that the scaling exponent does not improve for any finite number of helps from the genie. The proof technique is similar to that developed for $\rm MAP$ decoding and it is based on a DI bound.

{\bf \em Organization.} Section \ref{sec:mainres} states the DI bounds and their implications on the scaling exponent under $\rm MAP$ decoding with finite list size $L$. The proofs of these main results are contained in Section \ref{sec:proofBECL1} and \ref{sec:proofBMSC} for the $\rm BEC$ and for any $\rm BMSC$, respectively. The analysis concerning genie-aided decoding for transmission over the $\rm BEC$ is discussed in Section \ref{sec:furthres}. The conclusions of the paper are provided in Section \ref{sec:concl}.

\section{Main Results for MAP Decoding with List} \label{sec:mainres}

Let ${\mathcal C}_{\rm lin}$ be a set of linear codes parameterized
by their block length $N$ and rate $R$. For each $N$ and $R$, let
$d_{\rm min}(N, R)$ denote the minimum distance. Consider transmission over a $\rm BMSC$ $W$ with capacity $C \in (0, 1)$ and Bhattacharyya
parameter $Z \in (0, 1)$, defined as
\begin{equation}
Z(W) = \sum_{y\in\mathcal Y}\sqrt{W(y|0)W(y|1)},
\end{equation}
where $\mathcal Y$ denotes the output alphabet of the channel and $W(y|x)$ is the probability of receiving $y$ given that $x\in \{0, 1\}$ was transmitted. 
Let $P_e^{\rm
MAP}(N, R, W, L)$ be the block error probability for transmission
over $W$ under $\rm MAP$ decoding with list size
$L$. In addition, denote by ${\mathcal C}_{\rm pol}$ the set of
polar codes when transmission takes place over the $\rm BMSC$ $W$.

The case $W={\rm BEC}(\varepsilon)$ is handled separately. Indeed, for the $\rm BEC$,  $\rm MAP$ decoding reduces to solving a linear system over the finite field ${\mathbb F}_2$. Therefore, the number of codewords compatible with the received message is always a power of 2. Let us assume that the $\rm MAP$ decoder with list size $L$ declares an error if and only if the number of compatible codewords is strictly bigger than $L$. Consider a first $\rm MAP$ decoder with list size $L_1$, and a second $\rm MAP$ decoder whose list size $L_2$ is the biggest power of 2 smaller than $L_1$, i.e., $L_2 = 2^{\lfloor\log_2 L_1\rfloor}$. Then, the performance of these two decoders are identical (the first one declares error if and only if the second one does). As a result, we can restrict our analysis to list sizes which are powers of 2 and, therefore, the bounds can be tightened. In addition, when dealing with a $\rm BEC$, the DI approach itself and the proofs of the intermediate lemmas are considerably simpler, while keeping the same flavor as those valid for any general $\rm BMSC$.

\subsection{Divide and Intersect Bounds}
\label{subsec:boundBMS}

\begin{teo}[DI bound - ${\mathcal C}_{\rm lin}$] \label{teo:lplus1tolplingc}
Consider transmission using elements in ${\mathcal C}_{\rm lin}$ over a $\rm BMSC$ $W$ with Bhattacharyya parameter $Z$ and set $P_e \in (0, 1)$. For any $N$ and $R$ so that
\begin{align}
&P_e^{\rm MAP}(N, R, W, L) > P_e, \label{eq:condln1gc}\\
&d_{\rm min}(N, R) > \frac{\ln(P_e/8)}{\ln Z}, \label{eq:condln2gc}
\end{align}
the performance of the $\rm MAP$ decoder with list size $L+1$ ($2L$, if $W={\rm BEC}(\varepsilon)$) is lower bounded by
\begin{equation} \label{eq:boundgc}
\begin{split}
&P_e^{\rm MAP}(N, R, W, L+1) \ge \frac{3}{16} \cdot \bigl(P_e^{\rm MAP}(N, R, W, L)\bigr)^2,\\
&P_e^{\rm MAP}(N, R, \varepsilon, 2L) \ge \frac{3}{16} \cdot \bigl(P_e^{\rm MAP}(N, R, \varepsilon, L)\bigr)^2.
\end{split}
\end{equation}
\end{teo}

\begin{teo}[DI bound - ${\mathcal C}_{\rm pol}$] \label{teo:lplus1tolppolgc}
Consider transmission using elements in ${\mathcal C}_{\rm pol}$ over a $\rm BMSC$ $W$. Fix $P_e \in (0, 1)$ and pick any $N$ such that
\begin{equation} \label{eq:condNgc}
N > 2^{\bar{n}(Z, C, P_e)},
\end{equation}
where
\begin{equation} \label{eq:defnbargenLgc}
\begin{split}
&\bar{n}(Z, C, P_e) = 2\bar{m}(Z, P_e)  -\ln(1-C) \\
&\hspace{0.7cm}+ \sqrt{-4\bar{m}(Z, P_e)\cdot\ln(1-C)+ (\ln(1-C))^2}, 
\end{split}
\end{equation}
with
\begin{equation} \label{eq:defbarmgc}
\bar{m}(Z, P_e) = \log_2\Biggl(\frac{2\ln(P_e/8)\cdot \ln(1-Z)} {\ln Z\cdot \ln\biggl(1-Z^{\frac{4\ln(P_e/8)} {\ln Z}}\biggr)}\Biggr),
\end{equation}
and any sufficiently large $R$ so that
\begin{equation} \label{eq:condRgc}
P_e^{\rm MAP}(N, R, W, L) > P_e.
\end{equation}
Then, the bounds \eqref{eq:boundgc} hold.
\end{teo}

The following corollary follows by induction.

\begin{cor}[DI bound - any $L$] \label{cor:divintgenlgc}
Consider transmission using elements in ${\mathcal C}_{\rm lin}$ over a $\rm BMSC$ $W$. Fix $P_e \in (0, 1)$ and define the following recursion,
\begin{equation} \label{eq:srecgc}
P_e(m+1) = \frac{3}{16} \bigl(P_e(m)\bigr)^2, \qquad m \in {\mathbb N}\hspace{0.2em}, 
\end{equation}
with the initial condition $P_e(1) = P_e$. Pick any $N$ and $R$ such that \eqref{eq:condln1gc} and \eqref{eq:condln2gc} hold with $P_e(L)$ instead of $P_e$, or, if the code is in ${\mathcal C}_{\rm pol}$, any $N$ satisfying \eqref{eq:condNgc} and any sufficiently large $R$ satisfying \eqref{eq:condRgc}  with $P_e(L)$ instead of $P_e$. Then, the performance of the $\rm MAP$ decoder with list size $L+1$ is lower bounded by
\begin{equation} \label{eq:boundgcrec}
\begin{split}
&P_e^{\rm MAP}(N, R, W, L+1) \\
&\hspace{0.5cm} \ge \Bigl(\frac{3}{16}\Bigr)^{2^L-1} \cdot \bigl(P_e^{\rm MAP}(N, R, W, L=1)\bigr)^{2^L}.
\end{split}
\end{equation}
If $W={\rm BEC}(\varepsilon)$, consider the recursion \eqref{eq:srecgc} with the initial condition $P_e(0) = P_e$. If \eqref{eq:condln1gc}-\eqref{eq:condln2gc} and \eqref{eq:condNgc}-\eqref{eq:condRgc} are satisfied with $P_e(\log_2 L)$ instead of $P_e$ for codes in ${\mathcal C}_{\rm lin}$ and ${\mathcal C}_{\rm pol}$, respectively, then the performance of the $\rm MAP$ decoder with list size $2L$ is lower bounded by
\begin{equation} \label{eq:boundbecrc}
\begin{split}
&P_e^{\rm MAP}(N, R, \varepsilon, 2L) \\
&\hspace{0.8cm} \ge\Bigl(\frac{3}{16}\Bigr)^{2L-1} \cdot \bigl(P_e^{\rm MAP}(N, R, \varepsilon, L=1)\bigr)^{2L}.
\end{split}
\end{equation}
\end{cor}

\subsection{Scaling Exponent of \texorpdfstring{$\rm MAP$}{MAP} Decoding with List}

An immediate consequence of the DI bounds is that the scaling exponent defined in \eqref{eq:scal} does not change as long as $L$ is fixed and finite. More formally, one can define the existence of a scaling law as follows. 

\begin{defi}[Scaling law] \label{def:scal}
Consider a set of codes, parameterized by their block length $N$ and rate $R$, transmitted over the channel $W$, and processed by means of the decoder $\mathcal D$, and let $P_e^{\mathcal D}(N, R, W)$ denote the block error probability. We say that a \emph{scaling law} holds, if there exist a real number $\mu \in (0, +\infty)$, namely the \emph{scaling exponent}, and a function $f: {\mathbb R} \rightarrow [0, 1]$, namely the \emph{mother curve}, such that
\begin{equation}
\lim_{N\rightarrow\infty :\mbox{ }  N^{1/\mu}(C-R) = z} P_e^{\mathcal D}(N, R, W) =f(z).
\end{equation}
\end{defi}

The proof of the theorem below that bounds the scaling behavior of the $\rm MAP$ decoder with any finite list size $L$ is easily deduced from Corollary \ref{cor:divintgenlgc}. 

\begin{teo}[Scaling exponent - $\rm MAP$ decoding with list] \label{teo:scal}
Consider the set of polar codes ${\mathcal C}_{\rm pol}$ transmitted over a $\rm BMSC$ $W$. If the scaling law of Definition \ref{def:scal} holds for the $\rm MAP$ decoder with mother curve $f$ and scaling exponent $\mu$, then for any $L \in {\mathbb N}$,
\begin{equation}
\limsup_{N\rightarrow\infty :\mbox{ }  N^{1/\mu}(C-R) = z} P_e^{\rm MAP}(N, R, W, L) \le f(z),
\end{equation}
\begin{equation}
\begin{split}
&\liminf_{N\rightarrow\infty :\mbox{ }  N^{1/\mu}(C-R) = z} P_e^{\rm MAP}(N, R, W, L) \\
&\hspace{0.8cm}\ge \Bigl(\frac{3}{16}\Bigr)^{2^{L-1}-1} \cdot \bigl(f(z))^{2^{L-1}}.
\end{split}
\end{equation}
\end{teo}

In words, if a scaling law holds for the $\rm MAP$ decoder with list size $L$, the scaling exponent $\mu$ is the same as that for the original $\rm MAP$ decoder without list. Therefore, the speed at which capacity is approached as the block length grows large does not depend on the list size, provided that $L$ remains fixed. Notice that, in general, Theorem \ref{teo:scal} holds for any set of linear codes whose minimum distance is unbounded as the block length grows large.

\section{Proof of DI Bounds for MAP Decoding with List and \texorpdfstring{$W ={\rm BEC}(\varepsilon)$}{BEC}} \label{sec:proofBECL1}

As the name suggests, the DI procedure has two main ingredients: the \emph{Intersect} step is based on the correlation inequality stated in Section \ref{subsec:poscor}; the \emph{Divide} step is based on the existence of a suitable subset of codewords, which is discussed in Section \ref{subsec:exset}. The actual bound for linear and polar codes is proven for the simple case $L=1$ in Sections \ref{subsec:prooflin} and \ref{subsec:proofpol}, respectively, while the generalization to any list size is presented in Section \ref{subsec:genbec}.  

\subsection{Intersect Step: Correlation Inequality} \label{subsec:poscor}

Since the ${\rm BEC}$ is a symmetric channel, we can assume that the all-zero codeword has been transmitted. As the $\rm BEC$ does not introduce errors, the $\rm MAP$ decoder outputs all the codewords compatible with the received message. An error is declared if and only if the all-zero codeword is not the unique candidate, i.e., there is more than one candidate codeword.

Let us map the channel output into the erasure pattern $y = (y_1, \cdots, y_N) \in \{0, 1\}^N$, with $y_i = 1$ meaning that the $i$-th $\rm BEC$ has yielded an erasure symbol and $y_i =0$, otherwise. Let $G_y$ be the part of the generator matrix $G$ obtained by eliminating the columns corresponding to the erased symbols, i.e., all the columns of index $i$ s.t. $y_i =1$. It is easy to check that the $\rm MAP$ decoder outputs the information vector $u = (u_1, \cdots, u_{NR})$ if and only if $u G_y =0$. Define $E_u$ to be the set of all the erasure patterns such that $u$ solves $u G_y =0$, i.e.,
\begin{equation} \label{eq:defEu}
E_u = \{ y \in \{0, 1\}^N \hspace{0.2em}|\hspace{0.2em} u G_y = 0\}.
\end{equation}
Let $I_u$ be the set of positions $i$ in which $(u G)_i$ equals $1$, namely,
\begin{equation} \label{eq:defIu}
I_u = \{i \in \{1, \cdots, N\} \hspace{0.2em}|\hspace{0.2em} (uG)_i=1\}.
\end{equation}

Since $P_e^{\rm MAP}(N, R, \varepsilon, L=1)$ is the probability that there exists a non-zero informative vector $u$ that satisfies $u G_y =0$, we have
\begin{equation}
{\mathbb P}(\bigcup_{u \in U} E_u) = P_e^{\rm MAP}(N, R, \varepsilon, L=1),
\end{equation}
with
\begin{equation}
U = {\mathbb F}_2^{NR} \setminus 0^{NR},
\end{equation}
where $0^{NR}$ denotes a sequence of $NR$ $0$s.

We start with two simple lemmas computing ${\mathbb P}(E_u)$ and showing the positive correlation between the events \eqref{eq:defEu}.

\begin{lemma}[${\mathbb P}(E_u)$] \label{lm:pb}
Let $u\in {\mathbb F}_2^{NR}$ and let $E_u$ be defined in \eqref{eq:defEu}. Then,
\begin{equation} \label{eq:peu}
{\mathbb P}(E_u) = \varepsilon^{|I_u|},
\end{equation}
where $I_u$ is given by \eqref{eq:defIu}.
\end{lemma}

\begin{proof}
Observe that $u$ solves $uG_y =0$ if and only if all the positions $i$ s.t. $(u G)_i =1$ are erased by the ${\rm BEC}(\varepsilon)$. Therefore, ${\mathbb P}(E_u)$ equals the probability that $|I_u|$ independent erasures at those positions occur, which implies \eqref{eq:peu}.
\end{proof}

\begin{lemma}[Positive correlation between couples] \label{lm:pd}
Let $u, \tilde{u} \in {\mathbb F}_2^{NR}$. Then,
\begin{equation} \label{eq:pinter}
{\mathbb P}(E_u \cap E_{\tilde{u}}) \ge {\mathbb P}(E_u) \cdot {\mathbb P}(E_{\tilde{u}}).
\end{equation}
\end{lemma}

\begin{proof}
By definition \eqref{eq:defIu}, we obtain
\begin{equation}
\begin{split}
{\mathbb P}(E_u \cap E_{\tilde{u}}) &= \varepsilon^{|I_{u} \cup I_{\tilde{u}}|} = \varepsilon^{|I_{u}|+|I_{\tilde{u}}|-|I_{u} \cap I_{\tilde{u}}|} \\
& \ge \varepsilon^{{|I_{u}|+|I_{\tilde{u}}|}} = {\mathbb P}(E_u) \cdot {\mathbb P}(E_{\tilde{u}}),
\end{split}
\end{equation}
which gives \eqref{eq:pinter}.
\end{proof}

Let us now generalize Lemma \ref{lm:pd} to unions of sets.

\begin{lemma}[Positive correlation - ${\rm BEC}(\varepsilon)$, $L=1$] \label{lm:pdecunion}
Let $U_1, U_2 \subset {\mathbb F}_2^{NR}$. Then,
\begin{equation} \label{eq:pinterunion}
{\mathbb P}(\bigcup_{u \in U_1} E_u \cap \bigcup_{\tilde{u} \in U_2} E_{\tilde{u}}) \ge {\mathbb P}(\bigcup_{u \in U_1} E_u)\cdot {\mathbb P}(\bigcup_{\tilde{u} \in U_2} E_{\tilde{u}}).
\end{equation}
\end{lemma}

The proof of this result can be found in Appendix \ref{app:pdecunion} and comes from an application of the \emph{FKG inequality}, originally proposed in \cite{fkg:ineq}.

\subsection{Divide Step: Existence of a Suitable Subset of Codewords} \label{subsec:exset}

The aim of this section is to show that there exists $U_1 \subset U$ such that ${\mathbb P}(\bigcup_{u \in U_1} E_u)$ is slightly smaller than $\frac{1}{2} P_e^{\rm MAP}(N, R, \varepsilon, L=1)$. To do so, we first upper bound ${\mathbb P}(E_u)$ for all $u \in U$.

\begin{lemma}[No big jumps - ${\rm BEC}(\varepsilon)$] \label{lm:jump}
Let $P_e \in (0, 1)$ and $\varepsilon \in (0, 1)$. Then, for any $N$ and $R$ so that
\begin{equation} \label{eq:hpmin}
d_{\rm min}(N, R) > \frac{\ln(P_e/8)}{\ln \varepsilon},
\end{equation}
the probability of $E_u$ is bounded by
\begin{equation}
{\mathbb P}(E_u) < \frac{P_e}{8}, \qquad \forall \mbox{ }u \in U. 
\end{equation}  
\end{lemma}

\begin{proof}
From Lemma \ref{lm:pb} and the definition of minimum distance, we obtain that 
\begin{equation}
{\mathbb P}(E_u) = \varepsilon^{|I_u|} \le \varepsilon^{d_{\rm min}}.
\end{equation}
Using \eqref{eq:hpmin}, the thesis follows.
\end{proof}

The existence of a subset of codewords with the desired property is an immediate consequence of the previous lemma.

\begin{cor}[Existence of $U_1$] 
\label{cor:setcond}
Let $P_e \in (0, 1)$ and $\varepsilon \in (0, 1)$. Then, for any $N$ and $R$ so that \eqref{eq:hpmin} and $P_e^{\rm MAP}(N, R, \varepsilon, L=1) > P_e$ hold, there exists $U_1 \subset U$ which satisfies 
\begin{equation}\label{eq:setcond}
\begin{split}
{\mathbb P}(\bigcup_{u \in U_1} E_u)& \ge \frac{3}{8} P_e^{\rm MAP}(N, R, \varepsilon, L=1),\\
{\mathbb P}(\bigcup_{u \in U_1} E_u)& \le \frac{1}{2} P_e^{\rm MAP}(N, R, \varepsilon, L=1).
\end{split}
\end{equation}
\end{cor}

\subsection{Proof of Theorem \ref{teo:lplus1tolplingc}: DI Bound for Linear Codes} \label{subsec:prooflin}

At this point, we are ready to present the proof of Theorem \ref{teo:lplus1tolplingc} for the ${\rm BEC}$ and for a list size $L=1$. Recall that the Bhattacharyya parameter of a ${\rm BEC}(\varepsilon)$ is $Z=\varepsilon$.

\begin{proof}[Proof of Theorem \ref{teo:lplus1tolplingc} for ${\rm BEC}(\varepsilon)$, $L=1$]
Pick $U_1$ that satisfies \eqref{eq:setcond} and let $U_2 = U \setminus U_1$. Consequently,
\begin{equation*}
\begin{split}
&\frac{3}{8}\cdot P_e^{\rm MAP}(N, R, \varepsilon, L=1) \le {\mathbb P}(\bigcup_{u \in U_1} E_u),\\
&\frac{1}{2}\cdot P_e^{\rm MAP}(N, R, \varepsilon, L=1) \le {\mathbb P}(\bigcup_{\tilde{u} \in U_2} E_{\tilde{u}}).
\end{split}
\end{equation*}
Hence,
\begin{equation*}
\frac{3}{16} \cdot \bigl(P_e^{\rm MAP}(N, R, \varepsilon, L=1)\bigr)^2 \le {\mathbb P}(\bigcup_{u \in U_1} E_u) \cdot {\mathbb P}(\bigcup_{\tilde{u} \in U_2} E_{\tilde{u}}).
\end{equation*}

In addition, the following chain of inequalities holds,
\begin{equation*}
\begin{split}
&{\mathbb P}(\bigcup_{u \in U_1} E_u) \cdot {\mathbb P}(\bigcup_{\tilde{u} \in U_2} E_{\tilde{u}}) \le {\mathbb P}(\bigcup_{u \in U_1} E_u \cap \bigcup_{\tilde{u} \in U_2} E_{\tilde{u}})\\
& = {\mathbb P}(\bigcup_{u \in U_1, \tilde{u} \in U_2} E_u \cap E_{\tilde{u}}) \le {\mathbb P}(\bigcup_{u, \tilde{u} \in U, u \neq \tilde{u}} E_u \cap E_{\tilde{u}}),
\end{split}
\end{equation*}
where the first inequality comes from the application of Lemma \ref{lm:pdecunion} and the last passage is a direct consequence of $U_1 \cap U_2 = \emptyset$. Noticing that
\begin{equation*}
{\mathbb P}(\bigcup_{u, \tilde{u} \in U, u \neq \tilde{u}} E_u \cap E_{\tilde{u}}) =
P_e^{\rm MAP}(N, R, \varepsilon, L=2),
\end{equation*}
we obtain the desired result. 
\end{proof}

\subsection{Proof of Theorem \ref{teo:lplus1tolppolgc}: DI Bound for Polar Codes} \label{subsec:proofpol}

In order to apply the bound to polar codes, it suffices to prove the lower bound on the minimum distance, as required in \eqref{eq:hpmin}. 

\begin{lemma}[$d_{\rm min}$ of polar codes - ${\rm BEC}(\varepsilon)$] \label{lm:mindistpolar}
Consider a polar code in ${\mathcal C}_{\rm pol}$ for a ${\rm BEC}(\varepsilon)$. Let $P_e \in (0, 1)$, $\varepsilon \in (0, 1)$, and $N > 2^{\bar{n}(\varepsilon, P_e)}$, where
\begin{equation} \label{eq:defnbargenL}
\bar{n}(\varepsilon, P_e) = 2\bar{m}(\varepsilon, P_e)  -\ln \varepsilon  + \sqrt{-4\bar{m}(\varepsilon, P_e)\cdot\ln\varepsilon+ (\ln\varepsilon)^2}, 
\end{equation}
with
\begin{equation} \label{eq:defbarm}
\bar{m}(\varepsilon, P_e) = \log_2\Biggl(\frac{2\ln(P_e/8)\cdot \ln(1-\varepsilon)} {\ln\varepsilon\cdot \ln\biggl(1-\varepsilon^{\frac{2\ln(P_e/8)} {\ln\varepsilon}}\biggr)}\Biggr).
\end{equation}
Then, the lower bound on $d_{\rm min}$ \eqref{eq:hpmin} holds.
\end{lemma}

The proof of Lemma \ref{lm:mindistpolar} is in Appendix \ref{app:mindistpolar}. Thanks to this result, Theorem \ref{teo:lplus1tolppolgc} follows from Theorem \ref{teo:lplus1tolplingc}. Comparing \eqref{eq:defbarmgc} with \eqref{eq:defbarm}, we notice that, for the $\rm BEC$, the constraint on $N$ is less tight than the one required for any $\rm BMSC$.

\subsection{Generalization to Any List Size} \label{subsec:genbec}

Set $l= \log_2 L$ and define $E_{{\rm sp}(u^{(1)}, \cdots, u^{(l)})}$ to be the set of all erasure patterns $y$ such that the set of solutions of the linear system $u G_y =0$ contains the linear span generated by $\{u^{(1)}, \cdots, u^{(l)}\}$, i.e.,
\begin{equation} \label{eq:defEugenL}
\begin{split}
&E_{{\rm sp}(u^{(1)}, \cdots, u^{(l)})} = \bigcap_{u \in {\rm span}(u^{(1)}, \cdots, u^{(l)})} E_u \\
&= \{ y \in \{0, 1\}^N \hspace{0.2em}|\hspace{0.2em} u G_y = 0 \hspace{1.2em} \forall \mbox{ }u \in {\rm span}(u^{(1)}, \cdots, u^{(l)})\}.
\end{split}
\end{equation}
Consider the set ${\rm LS}_l$ containing all the linear spans of ${\mathbb F}_2^{NR}$ with $2^l$ elements. In formulae,
\begin{equation}
\begin{split}
&{\rm LS}_l = \{{\rm span}(u^{(1)}, \cdots, u^{(l)}) \hspace{0.2em}|\hspace{0.2em} u^{(i)} \in {\mathbb F}_2^{NR} \mbox{ }\forall i \in \{1, \cdots, l\},\\
&\hspace{2cm}|{\rm span}(u^{(1)}, \cdots, u^{(l)})| = 2^l\}.
\end{split}
\end{equation}

Since $P_e^{\rm MAP}(N, R, \varepsilon, L)$ is the probability that the solutions to the linear system $u G_y =0$ form a linear span of cardinality strictly greater than $L$, we have
\begin{equation}
\begin{split}
&{\mathbb P}(\bigcup_{{\rm span}(u^{(1)}, \cdots, u^{(l+1)}) \in {\rm LS}_{l+1}} E_{{\rm sp}(u^{(1)}, \cdots, u^{(l+1)})}) \\
&\hspace{4.5cm}= P_e^{\rm MAP}(N, R, \varepsilon, L).
\end{split}
\end{equation}

For the \emph{Intersect} step, we need now the generalization of Lemma \ref{lm:pdecunion}, which is contained in Lemma \ref{lm:pdecuniongen}.  

\begin{lemma}[Positive correlation - ${\rm BEC}(\varepsilon)$, any $L$] \label{lm:pdecuniongen}
Let $P_1, P_2 \subset {\rm LS}_{l}$. Then,
\begin{equation} \label{eq:pinterunionL}
\begin{split}
&{\mathbb P}(\bigcup_{{\rm span}(u^{(1)}, \cdots, u^{(l)}) \in P_1} E_{{\rm sp}(u^{(1)}, \cdots, u^{(l)})} \\
&\hspace{1cm}\cap \bigcup_{{\rm span}(\tilde{u}^{(1)}, \cdots, \tilde{u}^{(l)}) \in P_2} E_{{\rm sp}(\tilde{u}^{(1)}, \cdots, \tilde{u}^{(l)})}) \\
&\ge {\mathbb P}(\bigcup_{{\rm span}(u^{(1)}, \cdots, u^{(l)}) \in P_1} E_{{\rm sp}(u^{(1)}, \cdots, u^{(l)})})\\
&\hspace{0.2cm}\cdot {\mathbb P}(\bigcup_{{\rm span}(\tilde{u}^{(1)}, \cdots, \tilde{u}^{(l)}) \in P_2} E_{{\rm sp}(\tilde{u}^{(1)}, \cdots, \tilde{u}^{(l)})}).
\end{split}
\end{equation}
\end{lemma}

The proof of Lemma \ref{lm:pdecuniongen} is given in Appendix \ref{app:pdecuniongen}. We are going to need also the subsequent simple result concerning the intersection of events \eqref{eq:defEugenL}.

\begin{lemma}[Intersections] \label{lm:inters}
For any ${\rm span}(u^{(1)}, \cdots, u^{(l)})$ and ${\rm span}(\tilde{u}^{(1)}, \cdots, \tilde{u}^{(l)})$,
\begin{equation}
\begin{split}
&E_{{\rm sp}(u^{(1)}, \cdots, u^{(l)})}\cap E_{{\rm sp}(\tilde{u}^{(1)}, \cdots, \tilde{u}^{(l)})} \\
&\hspace{1.2cm}= E_{{\rm sp}(u^{(1)}, \cdots, u^{(l)}, \tilde{u}^{(1)}, \cdots, \tilde{u}^{(l)})}.
\end{split}
\end{equation}
\end{lemma}

\begin{proof}
Since ${\rm span}(u^{(1)}, \cdots, u^{(l)}, \tilde{u}^{(1)}, \cdots, \tilde{u}^{(l)}) \supset {\rm span}(u^{(1)}, \cdots, u^{(l)}) \cup {\rm span}(\tilde{u}^{(1)}, \cdots, \tilde{u}^{(l)})$,
\begin{equation*}
\begin{split}
&E_{{\rm sp}(u^{(1)}, \cdots, u^{(l)})}\cap E_{{\rm sp}(\tilde{u}^{(1)}, \cdots, \tilde{u}^{(l)})} \\
&\hspace{1.2cm}\supset E_{{\rm sp}(u^{(1)}, \cdots, u^{(l)}, \tilde{u}^{(1)}, \cdots, \tilde{u}^{(l)})}.
\end{split}
\end{equation*}
On the other hand, by linearity of the code, for any $u \in {\rm span}(u^{(1)}, \cdots, u^{(l)})$ and any $v \in {\rm span}(\tilde{u}^{(1)}, \cdots, \tilde{u}^{(l)})$, if $u G_y =0$ and $v G_y =0$, then $w G_y =0$ for all $w \in \{u+v : u \in {\rm span}(u^{(1)}, \cdots, u^{(l)}), v \in {\rm span}(\tilde{u}^{(1)}, \cdots, \tilde{u}^{(l)})\}$. As a result, 
\begin{equation*}
\begin{split}
&E_{{\rm sp}(u^{(1)}, \cdots, u^{(l)})}\cap E_{{\rm sp}(\tilde{u}^{(1)}, \cdots, \tilde{u}^{(l)})} \\
&\hspace{1.2cm}\subset E_{{\rm sp}(u^{(1)}, \cdots, u^{(l)}, \tilde{u}^{(1)}, \cdots, \tilde{u}^{(l)})},
\end{split}
\end{equation*}
and the thesis follows.
\end{proof}

As concerns the \emph{Divide} step, Corollary \ref{cor:setcondL} generalizes the result of Corollary \ref{cor:setcond} to any list size $L$. 

\begin{cor}[Existence of $P_1$] \label{cor:setcondL}
Let $P_e \in (0, 1)$ and $\varepsilon \in (0, 1)$. Then, for any $R$ and $N$ satisfying
\begin{align}
&P_e^{\rm MAP}(N, R, \varepsilon, L) > P_e,\\
&d_{\rm min} > \frac{\ln(P_e/8)}{\ln \varepsilon},
\end{align}
there exists $P_1 \subset {\rm LS}_{l+1}$ such that 
\begin{equation}\label{eq:setcondL}
\begin{split}
&{\mathbb P}(\bigcup_{{\rm span}(u^{(1)}, \cdots, u^{(l+1)}) \in P_1} E_{{\rm sp}(u^{(1)}, \cdots, u^{(l+1)})}) \\
&\hspace{4cm}\ge \frac{3}{8} P_e^{\rm MAP}(N, R, \varepsilon, L),\\
&{\mathbb P}(\bigcup_{{\rm span}(u^{(1)}, \cdots, u^{(l+1)}) \in P_1} E_{{\rm sp}(u^{(1)}, \cdots, u^{(l+1)})}) \\
&\hspace{4cm}\le \frac{1}{2} P_e^{\rm MAP}(N, R, \varepsilon, L).
\end{split}
\end{equation}
\end{cor}

At this point, we can prove Theorem \ref{teo:lplus1tolplingc} for the $\rm BEC$ and for any list size $L$. 

\begin{proof}[Proof of Theorem \ref{teo:lplus1tolplingc} for ${\rm BEC}(\varepsilon)$, any $L$]
Pick $P_1$ that satisfies \eqref{eq:setcondL} and let $P_2 = {\rm LS}_{l+1} \setminus P_1$. Consequently, applying Lemma \ref{lm:pdecuniongen} and \ref{lm:inters}, we have
\begin{equation*}
\begin{split}
\frac{3}{16} &\cdot \bigl(P_e^{\rm MAP}(N, R, \varepsilon, L)\bigr)^2 \\
&\le {\mathbb P}(\bigcup_{{\rm span}(u^{(1)}, \cdots, u^{(l+1)}) \in P_1} E_{{\rm sp}(u^{(1)}, \cdots, u^{(l+1)})}) \\
&\hspace{0.2cm}\cdot {\mathbb P}(\bigcup_{{\rm span}(\tilde{u}^{(1)}, \cdots, \tilde{u}^{(l+1)}) \in P_2} E_{{\rm sp}(\tilde{u}^{(1)}, \cdots, \tilde{u}^{(l+1)})})\\
& \le {\mathbb P}(\bigcup_{{\rm span}(u^{(1)}, \cdots, u^{(l+1)}) \in P_1} E_{{\rm sp}(u^{(1)}, \cdots, u^{(l+1)})} \\
&\hspace{0.4cm}\cap \bigcup_{{\rm span}(\tilde{u}^{(1)}, \cdots, \tilde{u}^{(l+1)}) \in P_2} E_{{\rm sp}(\tilde{u}^{(1)}, \cdots, \tilde{u}^{(l+1)})})\\
& \le {\mathbb P}(\bigcup_{\substack{{\rm span}(u^{(1)}, \cdots, u^{(l+1)}) \in P_1\\{\rm span}(\tilde{u}^{(1)}, \cdots, \tilde{u}^{(l+1)}) \in P_2}} E_{{\rm sp}(u^{(1)}, \cdots, u^{(l+1)}, \tilde{u}^{(1)}, \cdots, \tilde{u}^{(l+1)})})\\
& \le P_e^{\rm MAP}(N, R, \varepsilon, 2L),\\
\end{split}
\end{equation*}
where the last inequality is due to the fact that $|{\rm span}(u^{(1)}, \cdots, u^{(l+1)}, \tilde{u}^{(1)}, \cdots, \tilde{u}^{(l+1)})| \ge 2^{l+2} = 4L > 2L$, since $P_1 \cap P_2 = \emptyset$. 
\end{proof}

\section{Proof of DI Bounds for MAP Decoding with List and Any \texorpdfstring{$\rm BMSC$}{BMSC}} \label{sec:proofBMSC}

\subsection{Case \texorpdfstring{$L=1$}{L=1}}

Since the information vectors are equiprobable, the $\rm MAP$ decision rule is given by 
\begin{equation*}
\hat{u} = \arg\max_{\tilde{u}} p(y | \tilde{u}).
\end{equation*}

Define $E'_u$ as the set of all $y$ such that $p(y|u) \ge p(y|0^{NR})$. Simple algebraic manipulations show that 
\begin{equation} \label{eq:defEu1}
\begin{split}
E'_u &= \{y \in  {\mathcal Y}^N \hspace{0.2em}|\hspace{0.2em} \sum_{i=1}^N \ln \frac{p(y_i|(uG)_i)}{p(y_i|0)} \ge 0\} \\
&= \{y \in  {\mathcal Y}^N \hspace{0.2em}|\hspace{0.2em}\sum_{i \in I_u} \ln \frac{p(y_i|1)}{p(y_i|0)}\ge 0\}, 
\end{split}
\end{equation} 
where ${\mathcal Y}$ is the output alphabet of the channel and $I_u$ is defined in \eqref{eq:defIu}.

Note that $P_e^{\rm MAP}(N, R, \varepsilon, L=1)$ is the probability that there exists a non-zero informative vector $u$ s.t. $p(y|u) \ge p(y|0^{NR})$. Then, we have
\begin{equation}
{\mathbb P}(\bigcup_{u \in U} E'_u) = P_e^{\rm MAP}(N, R, W, L=1).
\end{equation}

As concerns the \emph{Intersect} part, we generalize the inequality of Lemma \ref{lm:pdecunion} with the correlation result of Lemma \ref{lm:pdecuniongc}.

\begin{lemma}[Positive correlation - $\rm BMSC$, $L=1$] \label{lm:pdecuniongc}
Let $U'_1, U'_2 \subset {\mathbb F}_2^{NR}$. Then,
\begin{equation} \label{eq:pinteruniongc}
{\mathbb P}(\bigcup_{u \in U'_1} E'_u \cap \bigcup_{\tilde{u} \in U'_2} E'_{\tilde{u}}) \ge {\mathbb P}(\bigcup_{u \in U'_1} E'_u)\cdot {\mathbb P}(\bigcup_{\tilde{u} \in U'_2} E'_{\tilde{u}}).
\end{equation}
\end{lemma}

The proof of Lemma \ref{lm:pdecuniongc} can be found in Appendix \ref{app:pdecuniongc}.

For the \emph{Divide} step, we need to show that ${\mathbb P}(E'_u)$ can be made as small as we want, as done in Lemma \ref{lm:jump} for the events \eqref{eq:defEu}. This result is provided by Lemma \ref{lm:jumpgc}, stated and proven below.

\begin{lemma}[No big jumps - $\rm BMSC$] \label{lm:jumpgc}
Let $P_e \in (0, 1)$ and $Z \in (0, 1)$. Then, for any $N$ and $R$ so that
\begin{equation} \label{eq:hpmingc}
d_{\rm min}(N, R) > \frac{\ln(P_e/8)}{\ln Z},
\end{equation}
the probability of $E'_u$ can be bounded as
\begin{equation}
{\mathbb P}(E'_u) < \frac{P_e}{8}, \qquad \forall \mbox{ }u \in U. 
\end{equation}  
\end{lemma}

\begin{proof}
It is possible to relate the probability of $E'_u$ and the Bhattacharyya parameter $Z$ of the $\rm BMSC$ $W$ as \cite[Lemma 4.66]{urbanke:coding} 
\begin{equation}
{\mathbb P}(E'_u) \le Z^{|I_u|}.
\end{equation}
Since $|I_u| \ge d_{\rm min} > \ln(P_e/8) / \ln Z$, the thesis easily follows.
\end{proof}

From Lemma \ref{lm:jumpgc} we deduce a result similar to that of Corollary \ref{cor:setcond}. Then, by using also Lemma \ref{lm:pdecuniongc} and by following the same procedure seen at the end of Section \ref{subsec:prooflin}, the proof of Theorem \ref{teo:lplus1tolplingc} with $L=1$ for any $\rm BMSC$ is readily obtained.

Lemma \ref{lm:mindistpolargc} generalizes the result of Lemma \ref{lm:mindistpolar}, showing that for $N$ big enough the required lower bound on the minimum distance holds. Hence, the DI bound and the subsequent scaling result are true for the class of polar codes.

\begin{lemma}[$d_{\rm min}$ of polar codes - $\rm BMSC$] \label{lm:mindistpolargc}
Let $P_e \in (0, 1)$, $Z \in (0, 1)$, and $N > 2^{\bar{n}(Z, C, P_e)}$, where $\bar{n}(Z, C, P_e)$ is given by \eqref{eq:defnbargenLgc}.
Then, the lower bound on $d_{\rm min}$ \eqref{eq:hpmingc} holds.
\end{lemma}

The proof of Lemma \ref{lm:mindistpolargc} is in Appendix \ref{app:mindistpolargc}.

\subsection{Generalization to Any List Size} \label{subsec:bmscL}

Let $E'_{u^{(1)}, \cdots, u^{(L)}}$ be the set of all $y$ such that $p(y|u) \ge p(y|0^{NR})$ for all $ u \in \{u^{(1)}, \cdots, u^{(L)}\}$, i.e.,
\begin{equation} \label{eq:defEugenLgc}
\begin{split}
&E'_{u^{(1)}, \cdots, u^{(L)}} = \bigcap_{u \in \{u^{(1)}, \cdots, u^{(L)}\}} E'_u =  \{y \in  {\mathcal Y}^N \hspace{0.2em}|\\
&\hspace{0.2cm} \sum_{i=1}^N \ln \frac{p(y_i|(uG)_i)}{p(y_i|0)} \ge 0\hspace{1.2em} \forall \mbox{ }u \in \{u^{(1)}, \cdots, u^{(L)}\}\}.
\end{split}
\end{equation}
Consider the set ${\rm SS}_L$ containing all the subsets of $L$ distinct elements of ${\mathbb F}_2^{NR}$. In formulae,
\begin{equation}
\begin{split}
{\rm SS}_L &= \{ \{u^{(1)}, \cdots, u^{(L)}\} : u^{(i)} \in {\mathbb F}_2^{NR}  \mbox{ }\forall i \in \{1, \cdots, L\},\\
&\hspace{1cm} u^{(i)}\neq u^{(j)} \hspace{0.6em} \forall \hspace{0.3em} i\neq j\}.
\end{split}
\end{equation}

Since $P_e^{\rm MAP}(N, R, W, L)$ is the probability that there are at least $L$ distinct non-zero information vectors $u^{(1)}, \cdots, u^{(L)}$ s.t. $p(y|u) \ge p(y|0^{NR})$ for all $ u \in \{u^{(1)}, \cdots, u^{(L)}\}$, we have
\begin{equation}
\begin{split}
{\mathbb P}(\bigcup_{\{u^{(1)}, \cdots, u^{(L)}\} \in {\rm SS}_{L}} E'_{u^{(1)}, \cdots, u^{(L)}})= P_e^{\rm MAP}(N, R, W, L).
\end{split}
\end{equation}

In order to prove Theorem \ref{teo:lplus1tolplingc} for any fixed list size $L$ and for any $\rm BMSC$, one can follow similar steps to those of Section \ref{subsec:genbec} and the result is readily obtained. The only part which requires some further investigation consists in the generalization of Lemma \ref{lm:pdecuniongc} with the result below, which is proven in Appendix \ref{app:pdecuniongengc}. 

\begin{lemma}[Positive correlation - $\rm BMSC$, any $L$] \label{lm:pdecuniongengc}
Let $P'_1, P'_2 \subset {\rm SS}_L$. Then,
\begin{equation} \label{eq:pinterunionLgc}
\begin{split}
&{\mathbb P}(\bigcup_{\{u^{(1)}, \cdots, u^{(L)}\} \in P'_1} E'_{u^{(1)}, \cdots, u^{(L)}} \\
&\hspace{1cm}\cap \bigcup_{\{\tilde{u}^{(1)}, \cdots, \tilde{u}^{(L)}\} \in P'_2} E'_{\tilde{u}^{(1)}, \cdots, \tilde{u}^{(L)}}) \\
&\ge {\mathbb P}(\bigcup_{\{u^{(1)}, \cdots, u^{(L)}\} \in P'_1} E'_{u^{(1)}, \cdots, u^{(L)}})\\
&\hspace{0.2cm}\cdot {\mathbb P}(\bigcup_{\{\tilde{u}^{(1)}, \cdots, \tilde{u}^{(L)}\} \in P'_2} E'_{\tilde{u}^{(1)}, \cdots, \tilde{u}^{(L)}}).
\end{split}
\end{equation}
\end{lemma}

\section{Further Results for Genie-Aided \texorpdfstring{$\rm SC$}{SC} Decoding and \texorpdfstring{$W={\rm BEC}(\varepsilon)$}{W=BEC}} \label{sec:furthres}

Consider a polar code in ${\mathcal C}_{\rm pol}$ transmitted over
a ${\rm BEC}(\varepsilon)$, and denote by $W_N^{(i)}$ the $i$-th synthetic
channel, which is a $\rm BEC$ of Bhattacharyya parameter $Z_i$. Let $P_e^{\rm SC}(N, R, \varepsilon, k)$ be the block error probability
under $\rm SC$ decoding aided by a $k$-genie. More precisely, a $k$-genie-aided $\rm SC$ decoder runs the usual $\rm SC$ algorithm with the following difference: when we
reach a synthetic channel associated to an information bit that is erased, namely, we cannot decide on the value of a certain information bit, the genie tells the value of the erased bit, and it does so a maximum of $k$ times. An error is declared if and only if the decoder requires more than $k$ helps from the genie. 

Consider now $\rm SCL$ decoding for transmission over the $\rm BEC$. The $\rm SCL$ decoder also runs the usual $\rm SC$ algorithm and, when we reach a synthetic channel  associated to an information bit that is erased, say $W_N^{(i)}$, it takes into account both the possibilities, namely $u_i=0$ and $u_i=1$, and it lets the two decoding paths evolve independently. In addition, when we reach a synthetic channel associated to a frozen bit, the $\rm SCL$ decoder gains new information, namely, it learns the value of the linear combination of some of the previous bits. Therefore, some of the existing decoding paths might not be compatible with this extra information and they are removed from the list. An error is declared if and only if at any point during the decoding process the decoder requires to store more than $L$ decoding paths. 

Note that a $k$-genie-aided $\rm SC$ decoder and an $\rm SCL$ decoder
with list size $2^k$ behave similarly but not identically. When we reach a synthetic channel associated to an information bit that is erased, the former uses one of the helps from the genie, and the
latter doubles the number of decoding paths. However, when we reach a synthetic channel associated to a frozen bit, the $\rm SCL$
decoder can reduce the number of decoding paths, while it is not possible to gain new helps from the genie. Therefore, the $\rm SCL$ decoder always succeeds when
the genie-aided decoder succeeds, but in addition it might also succeed in some cases where
the genie-aided decoder fails.

\subsection{DI Bound and Scaling Exponent} \label{sec:digenie}

\begin{teo}[DI bound - genie-aided decoding] \label{teo:kplus1tokgenie}
Consider the transmission of a polar code in ${\mathcal C}_{\rm pol}$ over a ${\rm BEC}(\varepsilon)$ and fix $P_e \in (0, 1)$. Pick $N$ big enough and any $R$ that ensures
\begin{equation} \label{eq:condgenie}
\frac{1}{4} > P_e^{\rm SC}(N, R, \varepsilon, k) > P_e.
\end{equation}
Then, the performance of the $k+1$-genie-aided SC decoder is lower bounded by
\begin{equation}
P_e^{\rm SC}(N, R, \varepsilon, k+1) \ge \frac{3}{16} \cdot \bigl(P_e^{\rm SC}(N, R, \varepsilon, k)\bigr)^2.
\end{equation}
\end{teo}

By induction, the corollary below easily follows.

\begin{cor}[DI bound - genie-aided decoding, any $k$] \label{cor:divintgenie}
Consider the transmission of a polar code in ${\mathcal C}_{\rm pol}$ over a ${\rm BEC}(\varepsilon)$. Fix $P_e \in (0, 1)$ and consider the recursion \eqref{eq:srecgc} with the initial condition $P_e(0) = P_e$. Pick $N$ big enough and $R$ such that \eqref{eq:condgenie} holds with $P_e(k)$ instead of $P_e$. Then, the performance of the $k+1$-genie-aided $\rm SC$ decoder is lower bounded by
\begin{equation}
\begin{split}
&P_e^{\rm SC}(N, R, \varepsilon, k+1) \\
&\hspace{0.5cm}\ge \Bigl(\frac{3}{16}\Bigr)^{2^{k+1}-1} \cdot \bigl(P_e^{\rm SC}(N, R, \varepsilon, k=0)\bigr)^{2^{k+1}}.
\end{split}
\end{equation}
\end{cor}

Roughly speaking, Theorem \ref{teo:kplus1tokgenie} implies that the scaling exponent cannot change under $\rm SC$ decoding for any fixed number of helps from the genie. This statement can be formalized as follows.

\begin{teo}[Scaling exponent - genie-aided decoding] \label{teo:scalgenie}
Consider the set of polar codes ${\mathcal C}_{\rm pol}$ transmitted over a ${\rm BEC}(\varepsilon)$. Assume that the scaling law of Definition \ref{def:scal} holds for the $\rm SC$ decoder with mother curve $f$ and scaling exponent $\mu$. Then, for any $k \in {\mathbb N}$,
\begin{equation}
\limsup_{N\rightarrow\infty :\mbox{ }  N^{1/\mu}(C-R) = z} P_e^{\rm SC}(N, R, \varepsilon, k) \le f(z),
\end{equation}
\begin{equation}
\begin{split}
&\liminf_{N\rightarrow\infty :\mbox{ }  N^{1/\mu}(C-R) = z} P_e^{\rm MAP}(N, R, \varepsilon, k) \\
&\hspace{0.8cm}\ge \Bigl(\frac{3}{16}\Bigr)^{2^k-1} \cdot \bigl(f(z))^{2^k}.
\end{split}
\end{equation}
\end{teo}

\subsection{Proof of DI Bound} \label{subsec:proofsgenie}

Let $y \in \{0, 1\}^N$ denote the erasure pattern of the channel and for any $i \in \{1, \cdots, N\}$ let $F_i$ be the set containing all $y$ such that $W_N^{(i)}$ erases, i.e.,
\begin{equation} \label{eq:defF}
F_i = \{y \in \{0, 1\}^N \hspace{0.2em}|\hspace{0.2em} W_N^{(i)} \mbox{ erases} \}.
\end{equation} 
Denoting by ${\mathcal F}^c$ the set of unfrozen positions, it is clear that
\begin{equation}
{\mathbb P}(\bigcup_{i \in {\mathcal F}^c} F_i)=P_e^{\rm SC}(N, R, \varepsilon, k=0).
\end{equation} 

The \emph{Intersect} step is based on the correlation inequality below, whose proof is similar to that of Lemma \ref{lm:pdecunion}. 

\begin{lemma}[Positive correlation for erasures - $k=0$] \label{lm:corerasure}
Let $I_1, I_2 \subset \{1, \cdots, N\}$. Then,
\begin{equation} \label{eq:eraseunion}
{\mathbb P}(\bigcup_{i \in I_1} F_i \cap \bigcup_{\tilde{\imath} \in I_2} F_{\tilde{\imath}}) \ge {\mathbb P}(\bigcup_{i \in I_1} F_i)\cdot {\mathbb P}(\bigcup_{\tilde{\imath} \in I_2} F_{\tilde{\imath}}).
\end{equation}
\end{lemma}

In general, define $F_{i_0, \cdots, i_k}$ to be the set of all the erasure patterns such that $W_N^{(i)}$ erases for all $i \in \{i_0, \cdots, i_k\}$, i.e., 
\begin{equation}
F_{i_0, \cdots, i_k} = \{y \in \{0, 1\}^N \hspace{0.2em}|\hspace{0.2em} W_N^{(i)} \mbox{ erases } \forall \hspace{0.3em}i \in \{i_0, \cdots, i_k\} \},
\end{equation}
and consider the set of positions ${\rm SP}_k$ containing all the subsets of $k$ distinct elements of ${\mathcal F}^c$,
\begin{equation}
\begin{split}
{\rm SP}_k &= \{ \{i_0, \cdots, i_k\} : i_m \in {\mathcal F}^c  \mbox{ }\forall m \in \{0, \cdots, k\},\\
&\hspace{1.2cm} i_m\neq i_n \hspace{0.6em} \forall \hspace{0.3em} m\neq n\}.
\end{split}
\end{equation}
It it clear that
\begin{equation}
{\mathbb P}(\bigcup_{\{i_0, \cdots, i_k\} \in {\rm SP}_k} F_{i_0, \cdots, i_k}) = P_e^{\rm SC}(N, R, \varepsilon, k).
\end{equation}
In addition, with a small effort we generalize the result of Lemma \ref{lm:corerasure} following the line of thought exposed in the proof of Lemma \ref{lm:pdecuniongen}.

\begin{lemma}[Positive correlation for erasures - any $k$] \label{lm:coerasuregen}
Let $R_1, R_2 \subset {\rm SP}_k$. Then,
\begin{equation} \label{eq:eraseunionk}
\begin{split}
&{\mathbb P}(\bigcup_{\{i_0, \cdots, i_k\} \in R_1} F_{i_0, \cdots, i_k} \cap \bigcup_{\{\tilde{\imath}_0, \cdots, \tilde{\imath}_k\} \in R_2} F_{\tilde{\imath}_0, \cdots, \tilde{\imath}_k}) \\
&\hspace{0.5cm}\ge {\mathbb P}(\bigcup_{\{i_0, \cdots, i_k\} \in R_1} F_{i_0, \cdots, i_k})\cdot {\mathbb P}(\bigcup_{\{\tilde{\imath}_0, \cdots, \tilde{\imath}_k\} \in R_2} F_{\tilde{\imath}_0, \cdots, \tilde{\imath}_k}).
\end{split}
\end{equation}
\end{lemma}

As concerns the \emph{Divide} step, we require the existence of $R_1 \subset {\rm SP}_k$, such that ${\mathbb P}(\bigcup_{\{i_0, \cdots, i_k\} \in R_1} F_{i_0, \cdots, i_k})$ is slightly less than $\frac{1}{2} P_e^{\rm SC}(N, R, \varepsilon, k)$. To prove this fact, we show that, choosing a suitably large block length, ${\mathbb P}(F_i)$ can be made as small as required. The proof of the lemma below is in Appendix \ref{app:jumperasure}.

\begin{lemma}[No big jumps for erasures] \label{lm:jumperasure}
Let $P_e \in (0, 1)$ and $\varepsilon \in (0, 1)$. Then, for $N\ge N_0(P_e, \varepsilon)$ and for any $R$ such that
\begin{equation} \label{eq:hplemmabigbhat}
\frac{1}{4} > P_e^{\rm SC}(N, R, \varepsilon, k=0) > P_e,
\end{equation}
the probability of $F_i$ is upper bounded by
\begin{equation}
{\mathbb P}(F_i) < \frac{P_e}{8}, \qquad \forall \hspace{0.3em} i \in {\mathcal F}^c. 
\end{equation}
\end{lemma}

\begin{cor}[Existence of $R_1$] \label{cor:setcondLerasure}
Let $P_e \in (0, 1)$ and $\varepsilon \in (0, 1)$. Then, for $N$ big enough and $R$ ensuring \eqref{eq:hplemmabigbhat}, there exists $R_1 \subset {\rm SP}_k$ such that 
\begin{equation}\label{eq:setconderase}
\begin{split}
{\mathbb P}(\bigcup_{\{i_0, \cdots, i_k\} \in R_1} F_{i_0, \cdots, i_k}) &\ge \frac{3}{8} P_e^{\rm SC}(N, R, \varepsilon, k),\\
{\mathbb P}(\bigcup_{\{i_0, \cdots, i_k\} \in R_1} F_{i_0, \cdots, i_k}) &\le \frac{1}{2} P_e^{\rm SC}(N, R, \varepsilon, k).
\end{split}
\end{equation}
\end{cor}

Eventually, the proof of Theorem \ref{teo:kplus1tokgenie} is obtained by using Lemma \ref{lm:coerasuregen} and Corollary \ref{cor:setcondLerasure} and by following a procedure similar to that outlined at the end of Section \ref{subsec:genbec}.

\section{Concluding Remarks} \label{sec:concl}

In this paper, the scaling exponent of list decoders is analyzed with an application to polar codes. By means of a \emph{Divide and Intersect} (DI) procedure, we lower bound the error probability under $\rm MAP$ decoding with list size $L$ for any $\rm BMSC$. The result applies to any set of linear codes with unbounded minimum distance, and, specifically, to the set of polar codes. As a result, we deduce that under $\rm MAP$ decoding the scaling exponent is a constant function of the list size.

A similar DI bound is proven for the genie-aided $\rm SC$ decoder, when transmission takes place over the $\rm BEC$. Consequently, the scaling exponent under genie-aided decoding does not change for any fixed number of helps from the genie.

\section*{Acknowledgement}

The authors would like to thank E. Telatar for fruitful discussions and the anonymous reviewers for their helpful comments. This work was supported by grant No. 200020\_146832/1 of the Swiss National Foundation.

\appendix

\subsection{Proof of Lemma \ref{lm:pdecunion}} \label{app:pdecunion}

\begin{proof}
Consider the Hamming space $\{0, 1\}^N$. For $y, z \in \{0, 1\}^N$ define the following partial order,
\begin{equation} \label{eq:partord}
y \le z \Longleftrightarrow y_i \le z_i, \qquad \forall i \in \{1, 2, \cdots, N\}.
\end{equation}  
Define $y \lor z$ and $y \land z$ as
\begin{equation} \label{eq:andor}
\begin{split}
(y \lor z)_i =& \begin{cases} 0 \qquad \mbox{ if } y_i = z_i =0,\\
1 \qquad\mbox{ else}, \end{cases}\\
(y \land z)_i =& \begin{cases} 1 \qquad \mbox{ if } y_i = z_i =1,\\
0 \qquad\mbox{ else}. \end{cases}\\
\end{split}
\end{equation}

Just to clarify the ideas, think of $y \in \{0, 1\}^N$ as an erasure pattern, as specified at the beginning of Section \ref{subsec:poscor}. Since the $N$ copies of the original ${\rm BEC}(\varepsilon)$ are independent and each of them is erased with probability $\varepsilon$, we consider the probability measure defined by
\begin{equation} \label{eq:probmeas}
{\mathbb P}(y) = \Bigl(\frac{\varepsilon}{1-\varepsilon}\Bigr)^{w_H(y)} (1-\varepsilon)^N,
\end{equation}
where $w_H$ denotes the Hamming weight.

As $w_H(y \lor z) + w_H(y \land z) = w_H(y) + w_H(z)$, we have 
\begin{equation} \label{eq:probmeascond}
{\mathbb P}(y) \cdot {\mathbb P}(z) = {\mathbb P}(y \lor z) \cdot {\mathbb P}(y \land z).
\end{equation}

For any $U_1 \subset {\mathbb F}_2^{NR}$, consider the function $f : \{0, 1\}^N \rightarrow \{0, 1\}$, defined as
\begin{equation*} 
f(y) = 1-\prod_{u \in U_1} (1-{\mathds 1}_{\{y \in E_u\}}), 
\end{equation*}
where $E_u$ is defined in \eqref{eq:defEu} and ${\mathds 1}_{\{y \in E_u\}} = 1$ if and only if $y \in E_u$. Consequently, if there exists $u \in U_1$ s.t. $u G_y =0$, then $f(y) = 1$; $f(y) = 0$, otherwise. Hence,
\begin{equation*}
{\mathbb E}[f(y)] = 1 \cdot {\mathbb P}(f(y) = 1) + 0 \cdot {\mathbb P}(f(y) = 0) = {\mathbb P}(\bigcup_{u \in U_1} E_u).
\end{equation*}
If $f(y) \le f(z)$ whenever $y \le z$, then $f$ is said to be monotonically increasing. If $y \le z$, then the erasure pattern $z$ contains all the erasures of $y$ (and perhaps some more). Thus, if $f(y)=1$, then $f(z)=1$.  Since $f$ can be either $0$ or $1$, this is enough to show that the function is increasing.

Analogously, for any $U_2 \subset {\mathbb F}_2^{NR}$, consider the function $g : \{0, 1\}^N \rightarrow \{0, 1\}$ defined as
\begin{equation*}
g(y) = 1-\prod_{\tilde{u} \in U_2} (1-{\mathds 1}_{\{y \in E_{\tilde{u}}\}}). 
\end{equation*}
The function $g$ is increasing and its expected value is given by 
\begin{equation*}
{\mathbb E}[g(y)] = {\mathbb P}(\bigcup_{\tilde{u} \in U_2} E_{\tilde{u}}).
\end{equation*}
In addition,
\begin{equation*}
{\mathbb E}[f(y)g(y)] = {\mathbb P}(\bigcup_{u \in U_1} E_u \cap \bigcup_{\tilde{u} \in U_2} E_{\tilde{u}}).
\end{equation*}
The thesis thus follows from the version of the FKG inequality presented in Lemma 40 of \cite{korada:exchange}.
\end{proof}

\subsection{Proof of Lemma \ref{lm:mindistpolar}} \label{app:mindistpolar}

\begin{proof}
Consider the $N \times N$ matrix $G_N$ defined as follows,
\begin{equation}
G_N = F^{\otimes n}, \qquad \qquad F =  \biggl[ \begin{array}{cc}
1 & 0 \\
1 & 1 \end{array} \biggr], 
\end{equation}
where $F^{\otimes n}$ denotes the $n$-th Kronecker power of $F$. Let $G = [g_1, g_2, \cdots, g_{NR}]^T$ be the generator matrix of the polar code of block length $N$ and rate $R$ for the ${\rm BEC}(\varepsilon)$, which is obtained by selecting the $NR$ rows of $G_N$ which minimize the corresponding Bhattacharyya parameters. Then, by Lemma 3 of \cite{hussami:perfo},
\begin{equation*}
d_{\rm min} = \min_{1\le i \le NR} w_H(g_i),
\end{equation*}
where $d_{\rm min}$ denotes the minimum distance.

Setting $n = \log_2 N$, we need to show that for $n > \bar{n}(\varepsilon, P_e)$,
\begin{equation} \label{eq:whinf}
w_H(g_i) > C(P_e, \varepsilon), \qquad i =1, 2, \cdots, NR,
\end{equation}
where
\begin{equation*}
C(P_e, \varepsilon) = \frac{\ln(P_e/8)}{\ln \varepsilon}.
\end{equation*}

Suppose, by contradiction, that \eqref{eq:whinf} does not hold, i.e., there exists a row $g_i$ s.t. for $n > \bar{n}(\varepsilon, P_e)$,
\begin{equation} \label{eq:limwh}
w_H(g_i) \le C(P_e, \varepsilon).
\end{equation}
Since $G$ is obtained from $G_N$ by eliminating the rows corresponding to the frozen indices, $g_i$ is a row of $G_N$, say row of index $i'$. Then, by Proposition 17 of \cite{arikan:polar},
\begin{equation*}
w_H(g_i) = 2^{w_H(b^{(i')})} = 2^{\sum_{j=1}^{n} b_j^{(i')}},
\end{equation*}
where $b^{(i')}=(b_1^{(i')}, b_2^{(i')}, \cdots, b_n^{(i')})$ is the binary expansion of $i'-1$ over $n$ bits, $b_1^{(i')}$ being the most significant bit and $b_n^{(i')}$ the least significant bit. Consequently, \eqref{eq:limwh} implies that
\begin{equation*}
\sum_{j=1}^{n} b_j^{(i')} \le \lceil\log_2 C(P_e, \varepsilon)\rceil = c(P_e, \varepsilon),
\end{equation*}
i.e., the number of 1's in the binary expansion of $i'-1$ is upper bounded by $c(P_e, \varepsilon)$. 

The Bhattacharyya parameter $Z_{i'}$ of the $i'$-th synthetic channel is given by
\begin{equation*}
Z_{i'} = f_{b_1^{(i')}} \circ f_{b_2^{(i')}} \circ \cdots f_{b_n^{(i')}} (\varepsilon),
\end{equation*}
where $\circ$ denotes function composition and 
\begin{align}
f_0(x) &= 1-(1-x)^2, \label{eq:f0}\\
f_1(x) &= x^2. \label{eq:f1}
\end{align}
Notice that $f_0$ and $f_1$ are increasing functions $\forall \mbox{ } x \in[0, 1]$, and that $f_1 \circ f_0(x) \ge f_0 \circ f_1(x)$ $\forall \mbox{ } x \in [0, 1]$. Consequently, if we set $m = w_H(b^{(i')})$, the minimum Bhattacharyya parameter $Z_{\rm min}(m)$ is obtained by applying first the function $f_1(x)$ $m$ times and then the function $f_0(x)$ $n-m$ times. The maximum Bhattacharyya parameter $Z_{\rm max}(m)$ is obtained if we apply first the function $f_0(x)$ $n-m$ times and then the function
$f_1(x)$ $m$ times. Observing also that for all $t \in {\mathbb N}$,
\begin{align}
&\underbrace{f_0 \circ f_0 \circ \cdots f_0(x)}_{\mbox{\emph{t} times}} = 1-(1-x)^{2^t}, \label{eq:f0it}\\ 
&\underbrace{f_1 \circ f_1 \circ \cdots f_1(x)}_{\mbox{\emph{t} times}} = x^{2^t}, \label{eq:f1it}
\end{align}
we get
\begin{equation}\label{eq:ineqbw}
Z_{\rm min}(m) \le Z_{i'} \le Z_{\rm max}(m),
\end{equation}
with
\begin{align*}
Z_{\rm min}(m) &= 1-(1-\varepsilon^{2^m})^{2^{n-m}},\\
Z_{\rm max}(m) &= (1-(1-\varepsilon)^{2^{n-m}})^{2^m}.
\end{align*}
Since $f_1(x) \le f_0(x)$ $\forall \mbox{ } x \in[0, 1]$ and $m \le c$, we obtain that
\begin{equation} \label{eq:ineqmc}
Z_{\rm min}(m) \ge Z_{\rm min}(c).
\end{equation}
At this point, we need to show that for $k$ sufficiently large,
\begin{equation} \label{eq:worstbest}
Z_{\rm min}(c) \ge Z_{\rm max}(c+k).
\end{equation}
As $1-(1-\varepsilon)^{2^{n-c-k}} < 1$, the condition \eqref{eq:worstbest} is satisfied if
\begin{equation*} 
1-(1-\varepsilon^{2^c})^{2^{n-c}} \ge 1-(1-\varepsilon)^{2^{n-k-c}},
\end{equation*}
which after some simplifications leads to
\begin{equation} \label{eq:k}
k \ge \log_2\biggl(\frac{\ln(1-\varepsilon)}{\ln(1- \varepsilon^{2^c})}\biggr).
\end{equation}

Notice that the RHS of \eqref{eq:k} is an increasing function of $c$. As $c < \log_2(C)+1$, we deduce that the choice
\begin{equation} \label{eq:fixk}
\bar{k} = \ceil[\Bigg]{\log_2\biggl(\frac{\ln(1-\varepsilon)}{\ln(1- \varepsilon^{2C})}\biggr)} = \ceil[\Bigg]{\log_2\Biggl(\frac{\ln(1-\varepsilon)}{\ln(1- \varepsilon^{\frac{2\ln(P_e/8)}{\ln \varepsilon}})}\Biggr)}
\end{equation}
also satisfies \eqref{eq:worstbest}.

An immediate consequence of inequalities \eqref{eq:ineqbw}, \eqref{eq:ineqmc}, and \eqref{eq:worstbest} is that $Z_{i'} \ge Z_{\rm max}(c+\bar{k})$. Therefore, we can conclude that every channel of index $j$ with $\ge c+\bar{k}$ ones in the binary expansion $b^{(j)}$ of $j-1$ has Bhattacharyya parameter $Z_j \le Z_{i'}$. Consequently, all these channels have not been frozen and, as $R \le C = 1-\varepsilon$,
\begin{equation*}
\begin{split}
\varepsilon \le 1-R &= \frac{\# \mbox{ frozen channels}}{\#\mbox{ channels}} \le \frac{\displaystyle\sum_{i=0}^{c+\bar{k}-1}\binom{n}{i}}{2^n} \\
&\le \exp{\Bigl(\frac{-(n-2(c+\bar{k}-1))^2}{2n}\Bigr)},
\end{split}
\end{equation*} 
where the last inequality is a consequence of Chernoff bound \cite{chernoff:article}.

After some calculations, we conclude that for $n > \bar{n}(\varepsilon, P_e)$, where $\bar{n}(\varepsilon, P_e)$ is given by \eqref{eq:defnbargenL},
\begin{equation*}
\exp{\Bigl(\frac{-(n-2(c+\bar{k}-1))^2}{2n}\Bigr)} < \varepsilon,
\end{equation*}
which is a contradiction.
\end{proof}

\subsection{Proof of Lemma \ref{lm:pdecuniongen}} \label{app:pdecuniongen}

\begin{proof}
As in the proof of Lemma \ref{lm:pdecunion} presented in Appendix \ref{app:pdecunion}, consider the Hamming space $\{0, 1\}^N$ with the partial order \eqref{eq:partord}. For $y, z \in \{0, 1\}^N$ define $y \lor z$ and $y \land z$ as in \eqref{eq:andor} and take the probability measure \eqref{eq:probmeas} which satisfies \eqref{eq:probmeascond}. For any $P_1, P_2 \subset {\rm LC}_l$, pick $f : \{0, 1\}^N \rightarrow \{0, 1\}$ and $g : \{0, 1\}^N \rightarrow \{0, 1\}$, defined as
\begin{equation*}
\begin{split}
f(y) &= 1-\prod_{{\rm span}(u^{(1)}, \cdots, u^{(l)}) \in P_1} \bigl(1-{\mathds 1}_{\{y \in E_{{\rm sp}(u^{(1)}, \cdots, u^{(l)})}\}}\bigr),\\ 
g(y) &= 1-\prod_{{\rm span}(\tilde{u}^{(1)}, \cdots, \tilde{u}^{(l)}) \in P_2} \bigl(1-{\mathds 1}_{\{y \in E_{{\rm sp}(\tilde{u}^{(1)}, \cdots, \tilde{u}^{(l)})}\}}\bigr),\\ 
\end{split}
\end{equation*}
where $E_{{\rm sp}(u^{(1)}, \cdots, u^{(l)})}$ is given by \eqref{eq:defEugenL} and  ${\mathds 1}_{\{y \in E_{{\rm sp}(u^{(1)}, \cdots, u^{(l)})}\}} = 1$ if and only if $y \in E_{{\rm sp}(u^{(1)}, \cdots, u^{(l)})}$. Hence,
\begin{equation*}
\begin{split}
&{\mathbb E}[f(y)] = {\mathbb P}(\bigcup_{{\rm span}(u^{(1)}, \cdots, u^{(l)}) \in P_1} E_{{\rm sp}(u^{(1)}, \cdots, u^{(l)})}),\\
&{\mathbb E}[g(y)] = {\mathbb P}(\bigcup_{{\rm span}(\tilde{u}^{(1)}, \cdots, \tilde{u}^{(l)}) \in P_2} E_{{\rm sp}(\tilde{u}^{(1)}, \cdots, \tilde{u}^{(l)})}),\\
&{\mathbb E}[f(y)g(y)] = {\mathbb P}(\bigcup_{{\rm span}(u^{(1)}, \cdots, u^{(l)}) \in P_1} E_{{\rm sp}(u^{(1)}, \cdots, u^{(l)})}\\
&\hspace{1cm}\cap \bigcup_{{\rm span}(\tilde{u}^{(1)}, \cdots, \tilde{u}^{(l)}) \in P_2} E_{{\rm sp}(\tilde{u}^{(1)}, \cdots, \tilde{u}^{(l)})}).\\
\end{split}
\end{equation*}
Since $f$ and $g$ are increasing, the thesis follows by Lemma 40 of \cite{korada:exchange}.
\end{proof}

\subsection{Proof of Lemma \ref{lm:pdecuniongc}} \label{app:pdecuniongc}

\begin{proof}
Assume for the moment that the output alphabet $\mathcal Y$ of the channel is finite and consider the binary relation $\stackrel{{\mathcal Y}}{\le}$, defined for all $y_i, z_i \in {\mathcal Y}$ as 
\begin{equation} \label{eq:defleA}
y_i \stackrel{{\mathcal Y}}{\le} z_i \Longleftrightarrow \frac{p(y_i|1)}{p(y_i|0)} \le \frac{p(z_i|1)}{p(z_i|0)}.
\end{equation}
The relation $\stackrel{{\mathcal Y}}{\le}$ is transitive and total. As concerns the antisymmetry, $\stackrel{{\mathcal Y}}{\le}$ satisfies the property if the following implication holds for all $y_i, z_i \in {\mathcal Y}$,
\begin{equation} \label{eq:llrorder}
\frac{p(y_i|1)}{p(y_i|0)} = \frac{p(z_i|1)}{p(z_i|0)} \Longrightarrow y_i = z_i.
\end{equation}
Note that, without loss of generality, we can assume that the channel output identifies with the log-likelihood ratio, see \cite[Section 4.1.2]{urbanke:coding}. With this assumption of using the canonical representation of the channel, \eqref{eq:llrorder} is also fulfilled. Hence, $\stackrel{{\mathcal Y}}{\le}$ is a total ordering over $\mathcal Y$.

Set ${\mathcal L} = {\mathcal Y}^N$ and for any $y = (y_1, \cdots, y_N)$ and $z =(z_1, \cdots, z_N)$ in ${\mathcal L}$ define the binary relation $\stackrel{{\mathcal L}}{\le}$ as
\begin{equation} \label{eq:defleL}
y \stackrel{{\mathcal L}}{\le} z \Longleftrightarrow y_i \stackrel{{\mathcal Y}}{\le} z_i, \hspace{1.5em} \forall \hspace{0.3em} i \in \{1, \cdots, N\}.
\end{equation}
It is easy to check that $\stackrel{{\mathcal L}}{\le}$ is a partial order over the $N$-fold Cartesian product ${\mathcal Y}^N$.

For any $y, z \in {\mathcal L}$, denote by $y \lor z$ their unique minimal upper bound and by $y \land z$ their unique maximal lower bound, defined as
\begin{align*}
(y \lor z)_i &= \max_{\stackrel{{\mathcal Y}}{\le}}(y_i, z_i), \qquad \forall \hspace{0.3em} i \in \{1, \cdots, N\},\\
(y \land z)_i &= \min_{\stackrel{{\mathcal Y}}{\le}}(y_i, z_i), \qquad \forall \hspace{0.3em} i \in \{1, \cdots, N\}. 
\end{align*}

Since the distributive law holds, i.e.,
\begin{equation*}
y \land (z \lor w) = (y \land z) \lor (y \land w), \qquad \forall \hspace{0.3em} y, z, w, \in {\mathcal L},
\end{equation*}
the set ${\mathcal L}$ with the partial ordering $\stackrel{{\mathcal L}}{\le}$ is a finite distributive lattice. Observe that in the proof of Appendix \ref{app:pdecunion} the finite distributive lattice ${\mathcal L}$ is replaced by the Hamming space $\{0, 1\}^N$.

Let $\mu : {\mathcal L} \rightarrow {\mathbb R}^{+}$ be defined as
\begin{equation} \label{eq:defmu}
\mu(y) = p(y|0^{NR}).
\end{equation}
In words, $\mu$ represents the probability of receiving the $N$-tuple $y$ from the channel, given that the all-zero information vector $0^{NR}$ was sent. We say that such a function is log-supermodular if, for all $y, z \in {\mathcal L}$,
\begin{equation} \label{eq:logsupermodular}
\mu(y)\cdot \mu(z) \le \mu(y \land z)\cdot \mu(y \lor z).
\end{equation}
An easy check shows that \eqref{eq:logsupermodular} is satisfied with equality with the choice \eqref{eq:defmu}. Notice that in the proof of Appendix \ref{app:pdecunion} the log-supermodular function $\mu$ is replaced by the probability measure \eqref{eq:probmeas}.

For any $U'_1 \subset {\mathbb F}_2^{NR}$, consider the function $f : {\mathcal L} \rightarrow \{0, 1\}$, defined as
\begin{equation*} 
f(y) = 1-\prod_{u \in U'_1} (1-{\mathds 1}_{\{y \in E'_u\}}), 
\end{equation*}
where $E'_u$ is given by \eqref{eq:defEu1} and ${\mathds 1}_{\{y \in E'_u\}} = 1$ if and only if $y \in E'_u$. If $f(y) \le f(z)$ whenever $y \stackrel{{\mathcal L}}{\le} z$, then $f$ is said to be monotonically increasing. Since $f$ can be either $0$ or $1$, we only need to prove the implication $f(y)=1 \Rightarrow f(z)=1$, whenever $y \stackrel{{\mathcal L}}{\le} z$. If $f(y) =1$, there exist $u^* \in U'_1$ such that
\begin{equation*}
0 \le \sum_{i \in I_{u^*}} \ln \frac{p(y_i|1)}{p(y_i|0)}. 
\end{equation*} 
As $y_i \stackrel{{\mathcal Y}}{\le} z_i$ for all $i \in \{1, \cdots, N\}$, by definition \eqref{eq:defleA} we obtain
\begin{equation*}
\sum_{i \in I_{u^*}} \ln \frac{p(y_i|1)}{p(y_i|0)} \le \sum_{i \in I_{u^*}} \ln \frac{p(z_i|1)}{p(z_i|0)}, 
\end{equation*}
which implies that $f(z) =1$. As a result, $f$ is increasing.

Analogously, for any $U'_2 \subset {\mathbb F}_2^{NR}$, consider the function $g : {\mathcal L} \rightarrow \{0, 1\}$ defined as
\begin{equation*}
g(y) = 1-\prod_{\tilde{u} \in U'_2} (1-{\mathds 1}_{\{ y \in E'_{\tilde{u}}\}}). 
\end{equation*}
Using the same argument seen for the function $f$, one realizes that $g$ is an increasing function.

By the FKG inequality \cite{alon:prob},
\begin{equation*}
\sum_{y \in {\mathcal L}} \mu(y) f(y) \cdot \sum_{y \in {\mathcal L}} \mu(y) g(y) \le \sum_{y \in {\mathcal L}} \mu(y) f(y)g(y) \cdot \sum_{y \in {\mathcal L}} \mu(y).
\end{equation*}
Observing that
\begin{align*}
&\sum_{y \in {\mathcal L}} \mu(y) =1,\\
&\sum_{y \in {\mathcal L}} \mu(y)f(y) ={\mathbb P}(\bigcup_{u \in U'_1} E'_u),\\
&\sum_{y \in {\mathcal L}} \mu(y)g(y) ={\mathbb P}(\bigcup_{\tilde{u} \in U'_2} E'_{\tilde{u}}),\\
&\sum_{y \in {\mathcal L}} \mu(y)f(y)g(y) ={\mathbb P}(\bigcup_{u \in U'_1} E'_u \cap \bigcup_{\tilde{u} \in U'_2} E'_{\tilde{u}}),
\end{align*}
we obtain the thesis \eqref{eq:pinteruniongc}.

When the output alphabet of the channel is infinite, the proof is very similar and follows from the generalization of the FKG inequality to a finite product of totally ordered measure spaces \cite{preston:fkggen}.

\end{proof}

\subsection{Proof of Lemma \ref{lm:mindistpolargc}} \label{app:mindistpolargc}

\begin{proof}
Following the approach of Appendix \ref{app:mindistpolar}, suppose, by contradiction, that there is an unfrozen index $i'$ of $G_N$, such that the number of 1's in the binary expansion of $i'-1$ is upper bounded by $c(P_e, Z)$, defined as
\begin{equation*}
c(P_e, Z) = \ceil[\Bigg]{\log_2 \frac{\ln(P_e/8)}{\ln Z}}.
\end{equation*} 
The Bhattacharyya parameter $Z_{i'}$ has the following expression
\begin{equation*}
Z_{i'} = f_{b_1^{(i')}} \circ f_{b_2^{(i')}} \circ \cdots f_{b_n^{(i')}} (Z),
\end{equation*}
where $f_1(x)$ is given by \eqref{eq:f1}, and $f_0(x)$ can be bounded as \cite{arikan:polar}, \cite[Problem 4.62]{urbanke:coding},
\begin{equation*}
\sqrt{1-(1-x^2)^2}= f_0^{(l)}(x) \le f_0(x) \le f_0^{(u)}(x) = 1-(1-x)^2.
\end{equation*}
Since $f_1(x)$ and $f_0^{(l)}(x)$ are increasing and $f_0^{(l)}(x) \le f_0(x)$, we have
\begin{equation*}
Z_{i'} \ge Z_{i'}^{(l)} = f_{b_1^{(i')}}^{(l)} \circ f_{b_2^{(i')}}^{(l)} \circ \cdots f_{b_n^{(i')}}^{(l)} (Z),
\end{equation*}
where, for the sake of simplicity, we have defined $f_1^{(l)}(x) = f_1(x)$. Setting $m = w_H(b^{(i')})$ and remarking that $f_1^{(l)} \circ f_0^{(l)}(x) \ge f_0^{(l)} \circ f_1^{(l)}(x)$, a lower bound on $Z_{i'}^{(l)}$ is obtained applying first the function $f_1^{(l)}(x)$ $m$ times and then the function $f_0^{(l)}(x)$ $n-m$ times. Using \eqref{eq:f1it} and observing that for all $t \in {\mathbb N}$,
\begin{equation*}
\underbrace{f_0^{(l)} \circ f_0^{(l)} \circ \cdots f_0^{(l)}(x)}_{\mbox{\emph{t} times}} = \sqrt{1-(1-x^2)^{2^t}}, 
\end{equation*}
we get
\begin{equation*}
Z_{i'}^{(l)} \ge Z_{\rm min}^{(l)}(m)=\sqrt{1-(1-Z^{2^{m+1}})^{2^{n-m}}}.
\end{equation*}
Since $f_1^{(l)}(x) \le f_0^{(l)}(x)$ and $m \le c$, we obtain that
\begin{equation*}
Z_{\rm min}^{(l)}(m) \ge Z_{\rm min}^{(l)}(c).
\end{equation*}

On the other hand, let $Z_j$ be the Bhattacharyya parameter of the synthetic channel of index $j$ with $\ge c+k$ ones in the binary expansion $b^{(j)}$ of $j-1$. Since $f_1(x)$ and $f_0^{(u)}(x)$ are increasing and $f_0(x) \le f_0^{(u)}(x)$, we have
\begin{equation*}
Z_j \le Z_j^{(u)} = f_{b_1^{(j)}}^{(u)} \circ f_{b_2^{(j)}}^{(u)} \circ \cdots f_{b_n^{(j)}}^{(u)} (Z),
\end{equation*}
where we have defined for the sake of simplicity $f_1^{(u)}(x) = f_1(x)$. Setting $m' = w_H(b^{(j)})$ and remarking that $f_1^{(u)} \circ f_0^{(u)}(x) \ge f_0^{(u)} \circ f_1^{(u)}(x)$, an upper bound on $Z_{j}^{(u)}$ is obtained applying first the function $f_0^{(u)}(x)$ $n-m'$ times and then the function $f_1^{(u)}(x)$ $m'$ times. Using \eqref{eq:f0it} and \eqref{eq:f1it}, we get
\begin{equation*}
Z_j^{(u)} \le Z_{\rm max}^{(u)}(m')=(1-(1-Z)^{2^{n-m'}})^{2^{m'}}.
\end{equation*}
Since $f_1^{(u)}(x) \le f_0^{(u)}(x)$ and $m' \ge c+k$, we have that
\begin{equation*}
Z_{\rm max}^{(u)}(m') \le Z_{\rm max}^{(u)}(c+k).
\end{equation*}

At this point, we need to pick $k$ such that the following inequality holds,
\begin{equation*} 
Z_{\rm min}^{(l)}(c) \ge Z_{\rm max}^{(u)}(c+k).
\end{equation*}
After some calculations, one obtains that
\begin{equation*}
\bar{k} = \ceil[\Bigg]{\log_2\Biggl(\frac{\ln(1-Z)}{\ln(1- Z^{\frac{4\ln(P_e/8)}{\ln Z}})}\Biggr)}
\end{equation*}
fulfills the requirement.
 
As a result, every channel of index $j$ with $\ge c+\bar{k}$ ones in the binary expansion $b^{(j)}$ of $j-1$ cannot be frozen. By Chernoff bound \cite{chernoff:article}, we get a contradiction for $n > \bar{n}(Z, C, P_e)$, where $\bar{n}(Z, C, P_e)$ is given by \eqref{eq:defnbargenLgc}.
\end{proof}

\subsection{Proof of Lemma \ref{lm:pdecuniongengc}} \label{app:pdecuniongengc}

\begin{proof}
Assume at first that the output alphabet $\mathcal Y$ of the channel is finite and consider the finite distributive lattice ${\mathcal L} = {\mathcal Y}^N$ with the partial ordering $\stackrel{{\mathcal L}}{\le}$ defined in \eqref{eq:defleL}. Let $\mu : {\mathcal L} \rightarrow {\mathbb R}^{+}$ be the log-supermodular function \eqref{eq:defmu}.

For any $P'_1, P'_2 \subset {\rm SS}_L$, consider the functions $f : {\mathcal L} \rightarrow \{0, 1\}$ and $g : {\mathcal L} \rightarrow \{0, 1\}$, given by
\begin{equation*}
\begin{split} 
f(y) &= 1-\prod_{\{u^{(1)}, \cdots, u^{(L)}\} \in P'_1} (1-{\mathds 1}_{\{y \in E'_{u^{(1)}, \cdots, u^{(L)}}\}}),\\
g(y) &= 1-\prod_{\{\tilde{u}^{(1)}, \cdots, \tilde{u}^{(L)}\} \in P'_2} (1-{\mathds 1}_{\{y \in E'_{\tilde{u}^{(1)}, \cdots, \tilde{u}^{(L)}}\}}),\\
\end{split} 
\end{equation*}
where $E'_{u^{(1)}, \cdots, u^{(L)}}$ is defined in \eqref{eq:defEugenLgc} and ${\mathds 1}_{\{y \in E'_{u^{(1)}, \cdots, u^{(L)}}\}} = 1$ if and only if $y \in E'_{u^{(1)}, \cdots, u^{(L)}}$. For analogous reasons to those pointed out in Appendix \ref{app:pdecuniongc}, $f$ and $g$ are monotonically increasing.

Noticing that
\begin{align*}
&\sum_{y \in {\mathcal L}} \mu(y)f(y) ={\mathbb P}(\bigcup_{\{u^{(1)}, \cdots, u^{(L)}\} \in P'_1} E'_{u^{(1)}, \cdots, u^{(L)}}),\\
&\sum_{y \in {\mathcal L}} \mu(y)g(y) ={\mathbb P}(\bigcup_{\{\tilde{u}^{(1)}, \cdots, \tilde{u}^{(L)}\} \in P'_2} E'_{\tilde{u}^{(1)}, \cdots, \tilde{u}^{(L)}}),\\
&\sum_{y \in {\mathcal L}} \mu(y)f(y)g(y) ={\mathbb P}(\bigcup_{\{u^{(1)}, \cdots, u^{(L)}\} \in P'_1} E'_{u^{(1)}, \cdots, u^{(L)}} \\
&\hspace{1cm}\cap \bigcup_{\{\tilde{u}^{(1)}, \cdots, \tilde{u}^{(L)}\} \in P'_2} E'_{\tilde{u}^{(1)}, \cdots, \tilde{u}^{(L)}}),
\end{align*}
the thesis follows from the FKG inequality \cite{alon:prob}. To handle the case of an infinite output alphabet, it is enough to apply the generalization of the FKG inequality in \cite{preston:fkggen}.

\end{proof}

\subsection{Proof of Lemma \ref{lm:jumperasure}} \label{app:jumperasure}

\begin{proof}
Suppose that the thesis does not hold, i.e.,
\begin{equation*}
\max_{i \in {\mathcal F}^c} {\mathbb P}(F_i) = \max_{i \in {\mathcal F}^c} Z_i = \alpha \ge \frac{P_e}{8}. 
\end{equation*}
Consider $a, b \in (0, 1)$ that satisfy
\begin{equation} \label{eq:condab}
\sqrt{a} \le 1-\sqrt{1-b}.
\end{equation}
Then, for any $\varepsilon \in (0, 1)$ and for $N$ sufficiently large, by Corollary 6 of \cite{HAU14} the number of channels $N_c(a, b, N, \varepsilon)$ whose Bhattacharyya parameter is contained in the interval $[a, b]$ is lower bounded by $N^{1+\lambda_{\rm BEC}^{(l)}}$, where $\lambda_{\rm BEC}^{(l)} \ge -0.279$. Since the choice $b=\alpha$ and $a= (\alpha/2)^2$ satisfies \eqref{eq:condab}, we obtain
\begin{equation} \label{eq:A}
N_c\biggl(\Bigl(\frac{\alpha}{2}\Bigr)^2, \alpha, N, \varepsilon\biggr) \ge A = \floor{N^{1+\lambda_{\rm BEC}^{(l)}}}.
\end{equation}
Let $B_i$ be the erasure indicator of the $i$-th synthetic channel of Bhattacharyya parameter $Z_i$. Then, $B_i \in \{0, 1\}$ is a binary random variable such that ${\mathbb P}(B_i=1)=Z_i$. Denote by $\rho_{i, j}$ the correlation coefficient between the erasure indicators of the $i$-th and the $j$-th channel, which can be expressed as
\begin{equation*}
\rho_{i, j} = \frac{{\mathbb E}(B_i B_j)-{\mathbb E}(B_i){\mathbb E}(B_j)}{{\rm var}(B_i){\rm var}(B_j)}.
\end{equation*}
By Corollary 2 of \cite{mani:correlation}, we have that
\begin{equation} \label{eq:sum}
\sum_{i, j \in \{1, \cdots, N\}} \rho_{i, j} \le N^{3-\log_2(3)}.
\end{equation}
Let ${\mathcal A}_{\rm max}$ be the set of indices of the unfrozen channels with the highest Bhattacharyya parameters such that $|{\mathcal A}_{\rm max}| = A$. Notice that the Bhattacharyya parameters of these channels are contained in the interval $[(\alpha/2)^2, \alpha]$ by \eqref{eq:A}. Denote by $R_A$ the associated $A\times A$ matrix of the correlation coefficients. We are going to show that for any $M \in {\mathbb N}$, there exists $S_M^* \subset {\mathcal A}_{\rm max}$, with $|S_M^*|=M$, such that    
\begin{equation} \label{eq:corrmax}
\max_{\substack{i, j \in S_M^*\\ i \neq j}} \rho_{i, j} < \binom{M}{2} \frac{N^{3-\log_2(3)}}{A^2}.
\end{equation}
Since $3-\log_2(3)-2(1+\lambda_{\rm BEC}^{(l)}) <0$, the previous relation implies that, if we fix $M$ and we choose $N$ suitably large, then the correlation coefficients of the channels with indices in $S_M^*$ can be made arbitrarily small.

To prove \eqref{eq:corrmax}, first observe that \eqref{eq:sum} clearly implies that $\sum_{i, j \in {\mathcal A}_{\rm max}} \rho_{i, j} \le N^{3-\log_2(3)}$. Hence, the average of all the elements of the matrix $R_A$ is upper bounded by $N^{3-\log_2(3)}/A^2$. As $R_A$ is symmetric and its principal diagonal is made up by ones, the average of the strictly upper triangular part of $R_A$, namely the average of the $\binom{A}{2}$ elements of $R_A$ which are above the principal diagonal, is also upper bounded by $N^{3-\log_2(3)}/A^2$. In formulae,
\begin{equation*}
\frac{1}{\binom{A}{2}}\sum_{\substack{i, j \in {\mathcal A}_{\rm max}\\ i < j}} \rho_{i, j} \le \frac{N^{3-\log_2(3)}}{A^2}.
\end{equation*}
To any $S_M \subset {\mathcal A}_{\rm max}$, with $|S_M|=M$, we can associate the $\binom{M}{2}$ elements of the strictly upper triangular part of $R_A$ which represent the correlation coefficients of the channels whose indices are in $S_M$. By symmetry, when we consider all the subsets of cardinality $M$ of ${\mathcal A}_{\rm max}$, we count each element of the strictly upper triangular part of $R_A$ the same number of times, i.e. $\binom{A-2}{M-2}$. As a result, noticing that there are $\binom{A}{M}$ distinct subsets of cardinality $M$ of ${\mathcal A}_{\rm max}$, we have
\begin{equation*}
\frac{1}{\binom{A}{M}}\sum_{S_M \subset {\mathcal A}_{\rm max}} \frac{1}{\binom{M}{2}}\sum_{\substack{i, j \in S_M\\ i < j}} \rho_{i, j} \le \frac{N^{3-\log_2(3)}}{A^2}.
\end{equation*}
Consequently, there exists $S_M^* \subset {\mathcal A}_{\rm max}$, such that
\begin{equation*}
\frac{1}{\binom{M}{2}}\sum_{\substack{i, j \in S_M^*\\ i < j}} \rho_{i, j} \le \frac{N^{3-\log_2(3)}}{A^2},
\end{equation*}
which implies \eqref{eq:corrmax}.

With the choice $M = \ceil{128/P_e^2}$, it is easy to see that there exists $S^* \subset S_M^*$ that satisfies
\begin{equation}\label{eq:sumbhatlim}
\frac{1}{2}+\alpha\ge \sum_{i \in S^*} Z_i \ge \frac{1}{2}. 
\end{equation} 
Indeed, $\sum_{i \in S_M^*}Z_i \ge M(\alpha/2)^2 \ge 1/2$ and $\max_{i \in S_M^*} Z_i \le \alpha$. 

An application of Bonferroni's inequality (see \cite[Section 4.7]{comtet:comb}) yields
\begin{equation}\label{eq:calc1}
\begin{split}
P_e^{\rm SC}(N, R, \varepsilon, k=0)&\ge {\mathbb P}(\bigcup_{i \in S^*} F_i)\\
&\ge \sum_{i \in S^*} {\mathbb P}(F_i) - \frac{1}{2} \sum_{\substack{i, j \in S_M^*\\ i \neq j}}{\mathbb P}(F_i \cap F_j).
\end{split}
\end{equation}
The term ${\mathbb P}(F_i \cap F_j)$ can be upper bounded as
\begin{equation} \label{eq:calc2}
\begin{split}
{\mathbb P}(F_i \cap F_j) &= Z_i Z_j +\rho_{i, j} \sqrt{Z_i Z_j(1-Z_i)(1-Z_j)} \\
& \le Z_i Z_j + \binom{M}{2}\frac{N^{3-\log_2(3)}}{A^2} \\
& \le Z_i Z_j + \frac{1}{8\binom{M}{2}},
\end{split}
\end{equation}
where the first inequality comes from \eqref{eq:corrmax} and the fact that $Z_i \in [0, 1]$ and the second inequality is easily obtained picking $N$ large enough.

Using \eqref{eq:calc1} and \eqref{eq:calc2}, we have
\begin{equation}
P_e^{\rm SC}(N, R, \varepsilon, k=0) \ge \sum_{i \in S^*} Z_i - \frac{1}{2} \Bigl(\sum_{i \in S^*} Z_i\Bigr)^2 - \frac{1}{8}.
\end{equation}
Note that
\begin{equation*}
\alpha \le P_e^{\rm SC}(N, R, \varepsilon, k)<\frac{1}{4},
\end{equation*}
where the last inequality comes from the hypothesis of the Lemma. Hence, by using \eqref{eq:sumbhatlim}, we deduce that 
\begin{equation*}
\sum_{i \in S^*} Z_i < \frac{3}{4} < 1.
\end{equation*}
Since the function $h(x) = x-x^2/2$ is increasing in $[0, 1]$ and $1/2\le \sum_{i \in S^*} Z_i < 1$, we can conclude that
\begin{equation}
\sum_{i \in S^*} Z_i - \frac{1}{2} \Bigl(\sum_{i \in S^*} Z_i\Bigr)^2 - \frac{1}{8}\ge \frac{1}{4},
\end{equation}
which is a contradiction and gives us the thesis.

\end{proof}

\bibliographystyle{IEEEtran}

\bibliography{biblio}

\end{document}